\documentclass[11pt,a4paper]{article}
\pdfoutput=1
\usepackage[margin=28mm,headheight=14pt]{geometry}
\usepackage[T1]{fontenc}
\usepackage[utf8]{inputenc}
\usepackage{lmodern}
\usepackage{amsmath,amssymb,amsthm,mathtools,mathrsfs}
\usepackage{booktabs,array,longtable,microtype}
\usepackage[numbers,sort&compress]{natbib}
\usepackage{xcolor}
\usepackage[colorlinks=true,linkcolor=blue,citecolor=blue,urlcolor=blue]{hyperref}
\usepackage{fancyhdr}
\numberwithin{equation}{section}
\allowdisplaybreaks[2]
\newcommand{\h}{\hbar}
\newcommand{\dd}{\mathrm d}
\newcommand{\tr}{\operatorname{tr}}
\newcommand{\Res}{\operatorname*{Res}}
\newcommand{\remd}{\operatorname{rem}}
\newcommand{\Rhat}{\widehat R}
\newcommand{\Qhat}{\widehat Q}
\newcommand{\cM}{\mathcal M}
\newcommand{\cC}{\mathcal C}
\newcommand{\cK}{\mathcal K}
\newcommand{\cD}{\mathcal D}
\newcommand{\cP}{\mathcal P}

\newcommand{\ii}{\mathrm i}

\newtheorem{proposition}{Proposition}[section]
\newtheorem{theorem}[proposition]{Theorem}
\newtheorem{lemma}[proposition]{Lemma}
\newtheorem*{localpolelemma}{Lemma}

\hypersetup{pdftitle={Resolvent reconstruction of minimal strings beyond KdV},pdfauthor={Wasif Ahmed and Joydeep Naskar}}
\begin{document}
\hypersetup{pageanchor=false}
\begin{titlepage}
\noindent\hfill
\vspace{2.1cm}
\begin{flushleft}
{\LARGE\sffamily\bfseries Resolvent reconstruction of minimal strings\\[0.25em]
beyond KdV\\[0.25em]\par}
\end{flushleft}
\vspace{0.5cm}\hrule\vspace{0.65cm}
\noindent\textbf{Wasif Ahmed}$^{a}$,\quad\textbf{Joydeep Naskar}$^{b}$
\vspace{0.3cm}

\noindent{\small\itshape
$^a$Physics Department, Broida Hall, University of California, Santa Barbara, CA 93106, USA\par
\noindent$^b$Beijing Institute of Mathematical Sciences and Applications, Huaibei Town, Huairou District, Beijing 101408, China\par}
\vspace{0.15cm}
\noindent\textit{E-mail:}\ \href{mailto:wasif@ucsb.edu}{\texttt{wasif@ucsb.edu}},\quad
\href{mailto:joydeepnaskar@bimsa.cn}{\texttt{joydeepnaskar@bimsa.cn}}
\vspace{0.55cm}

\noindent\textsc{Abstract:}
Minimal strings offer two ways to study quantum surfaces: through a spectral curve or through a differential equation. We ask how much of the quantum equation can be recovered once the classical curve and its motion with the background are known. Our starting point is the resolvent of a scalar differential operator. We require its one-boundary integral to have poles only at branch points, with the resulting differential regular at infinity. These conditions give a finite reconstruction procedure at each genus. When the branch points are simple and their energy values move, we prove that the quantum corrections are unique whenever the procedure is consistent. A specified string equation provides a solution if its trace primitive satisfies the same analytic conditions. The scalar construction also leads to amplitudes with several boundaries. We introduce the method through pure gravity and Ising, then explore higher-order, dual and nonunitary models. Higher-Airy examples connect the calculation to $r$-spin theory on a degenerate locus that requires a separate treatment of the inverse problem.
\vspace{0.4cm}

\noindent\textbf{Keywords:} Matrix Models, Two-Dimensional Gravity, Integrable Hierarchies, Minimal Strings, Topological Recursion
\end{titlepage}
\setcounter{page}{1}
\hypersetup{pageanchor=true}
\setcounter{tocdepth}{2}
{\small\tableofcontents}
\clearpage

\section{Introduction}\label{sec:intro}
A surface can change in two simple ways: we can add a handle, or we can open another boundary. Minimal strings give a setting in which both operations can be studied systematically. They describe two-dimensional gravity coupled to minimal-model matter, and their matrix-model formulation organizes the sum over surfaces by genus and number of boundaries \cite{BrezinKazakov1990,DouglasShenker1990,GrossMigdal1990,DiFrancescoGinspargZinnJustin1995}. The disk is the starting point. Its amplitude is encoded in a spectral curve, from which loop equations and topological recursion construct the higher-genus, multi-boundary amplitudes \cite{EynardOrantin2007}.

There is also a differential-equation description of the same perturbative problem \cite{Douglas1990,BanksDouglasSeibergShenker1990}. Here a string equation selects a background, and the resolvent of an auxiliary differential operator gives access to its observables. For one boundary, the resolvent can be integrated at fixed spectral energy and evaluated at an endpoint. This provides a particularly direct route to the amplitude. In topological recursion, the calculation of a higher-genus one-boundary function generally passes through functions with more boundaries; the scalar ODE allows us to follow the one-boundary problem on its own.

The familiar KdV example makes the two steps easy to see. A string equation first determines a potential $u(x)$. The Gel'fand--Dikii equation then determines the diagonal resolvent of the associated Schr\"odinger operator \cite{GelfandDikii1975,Johnson2024Virasoro,AhmedJohnsonSaraswat2025}. We would like to turn part of this reasoning around. Suppose the leading potential is given and we want to find its quantum corrections. Can we recover those corrections by asking that the integrated resolvent have the analytic behavior of a boundary amplitude? And can the same idea work when the differential operator contains several unknown fields?

The one-boundary ODE method provides a useful starting point for this question. Johnson developed it for the Virasoro minimal string and its supersymmetric relatives \cite{Johnson2024Virasoro,Johnson2024SuperVirasoro}. Ahmed, Johnson and Saraswat developed the genus expansion and a descending construction of resolvent primitives, with applications to extended JT supergravity \cite{AhmedJohnsonSaraswat2025}. In their second-order construction, matching the poles from strongest to weakest also determines the next potential correction through the final coefficient equation. This gives a direct precedent for the reconstruction we develop here with several unknown fields. The role of the chosen string solution is also visible in comparisons of supersymmetric backgrounds and their loop observables \cite{JohnsonKolanowski2026,Johnson2025Further}.

Several-boundary amplitudes have a well-developed KdV description. Okuyama and Sakai organize macroscopic loops by KdV flows \cite{OkuyamaSakai2019}, and Lowenstein studies loop and brane correlators on general backgrounds, including closed expressions at fixed low genus \cite{Lowenstein2024OpenClosed,Lowenstein2024Volumes}. Johnson's universal $W_{g,n}$ formulas express these observables in terms of a single leading potential and its derivatives \cite{Johnson2026Universal}. At higher order, the corresponding differential systems and determinantal correlation functions are available for Drinfeld--Sokolov hierarchies and reduced KP \cite{BertolaDubrovinYang2015,BertolaDubrovinYang2016,BergereEynard2009,BergereBorotEynard2015}. Related approaches derive equations for correlators of specified integrable systems \cite{EynardMitsiosOukassi2026} and describe macroscopic-loop insertions through open Gel'fand--Dikii hierarchies \cite{AleshkinBelavin2019}. We will use these connections to attach additional boundaries to the scalar construction.

How does a classical spectral curve lead to a quantum differential equation? Bouchard and Eynard reconstruct WKB wavefunctions and quantum curves through topological recursion for a large class of curves, including higher-order examples \cite{Bouchard:2016obz}. Eynard, Garcia-Failde, Marchal and Orantin extend this construction to generic spectral curves with simple ramification \cite{Eynard:2021sxg}. These constructions involve choices: the integration divisor affects operator ordering, while shifted topological recursion produces families of quantum curves with additional parameters \cite{Bouchard:2016obz,Belliard:2024pae}.

Our reconstruction asks for the quantum fields of the spatial operator along a prescribed moving classical family, with its ordering and allowed genus expansion fixed. At third order, for example, the classical energy has the form $E=p^3+bp+a$. Both $a$ and $b$ can vary with the background coordinate. If only their leading terms are known, the resolvent can still be expanded, with the higher-genus corrections to both fields left as unknowns. We will see that the allowed poles of its integral can determine those unknowns. This two-field example is the simplest setting in which to understand the higher-order reconstruction.

To state the problem more precisely, consider the principal $A_{r-1}$ Drinfeld--Sokolov hierarchy, equivalently the $r$-reduction of KP \cite{DrinfeldSokolov1985,Dubrovin1996}. The integer $r$ is the scalar operator order and the dimension of the companion representation; the Lie algebra $A_{r-1}$ has rank $r-1$. With $D=\h\partial_x$, we choose the symmetric ordering
\begin{equation}\label{eq:introQ}
 Q=D^r+\sum_{j=0}^{r-2}\frac12\{u_j,D^j\},
 \qquad u_j(x;\h)=\sum_{g\geq0}\h^{2g}u_{j,g}(x).
\end{equation}
The anticommutator specifies how derivatives and coefficient fields are placed in the operator. Commuting a derivative past a field produces additional terms, whose coefficients depend on this choice of ordering. The leading approximation to $Q\psi=E\psi$ is the algebraic relation
\begin{equation}\label{eq:introE}
 E(p,x)=p^r+\sum_{j=0}^{r-2}u_{j,0}(x)p^j,
 \qquad \Delta(p,x)=\partial_pE(p,x).
\end{equation}
There are three variables to keep track of. The coordinate $x$ describes the background, $p$ parametrizes the classical curve, and $E$ is the spectral energy. At a generic fixed energy there are $r$ values of $p$, or $r$ sheets. Two local sheets meet at a simple zero of $\Delta$.

Knowing the polynomial $E(p,\mu)$ at an endpoint is only part of the story. We also need to know how its coefficients vary with $x$ and which disk function $y(p)$ this motion selects. Pure gravity and Ising make the distinction concrete: they can share an endpoint energy polynomial while having different disks. The input to our inverse problem is therefore a leading \emph{family} of curves, including its background motion.

The scalar equation and its adjoint come with a conserved bilinear pairing, the higher-order counterpart of a Wronskian. Its normalization gives a formal diagonal resolvent $\Rhat$ whose leading term is $1/\Delta$. At each positive genus we seek a rational primitive $\Qhat_g$:
\begin{equation}\label{eq:introprimitive}
 \left.\partial_x\right|_E\Qhat_g=\Rhat_{2g},
 \qquad \omega_{g,1}=\Qhat_g\,\dd E\big|_{x=\mu}.
\end{equation}
The derivative is taken at fixed energy. As the background changes, $p$ must move along with the fields to keep $E$ unchanged. We allow finite poles of $\Qhat_g$ only where sheets meet, and require the amplitude differential to be regular at infinity. These two analytic conditions turn the search for a primitive and its unknown fields into a finite calculation.

The key observation is that the most singular terms are already known before the new fields are determined. At genus $g$, changing those fields changes the resolvent only through poles of order at most three. The fixed-energy derivative of a nonzero allowed primitive has a pole of order at least four, provided the branch points are simple and their energy values move. We can therefore start with the strongest poles and work downward: first determine the primitive, then recover the fields from what remains. This proves uniqueness at every order $r$ and genus $g$ whenever the compatibility conditions hold. The compatibility conditions decide whether a solution exists for the proposed leading family.

A known string equation helps us understand when a solution exists. For a pair satisfying $[P,Q]=\kappa\h$, the second operator defines a spectral connection $\mathsf B$. The same bilinear pairing that gave the resolvent also gives a sheet projector $\cM$, and the two obey
\begin{equation}\label{eq:introtrace}
 \partial_x\tr(\mathsf B\cM)=\Rhat.
\end{equation}
Thus the string pair supplies a primitive. If its positive-genus coefficients satisfy our pole and infinity conditions, the uniqueness result forces the inverse construction to recover that string solution. The projector also leads to connected amplitudes with several boundaries through cyclic traces. Loop equations for differential systems on curves provide a general setting for these correlators \cite{Belliard:2016mxj}. Their identification with topological recursion uses the topological-type property of the differential system \cite{BergereBorotEynard2015}; sufficient conditions for this property are established for integrable Lax systems in \cite{Belliard:2016cyc}.

The main symbols are collected in Table~\ref{tab:notation}. We will introduce the remaining notation when it is needed; in particular, the distinction between the operator $Q$ and the primitive $\Qhat_g$ will be used throughout.

\begin{table}[htbp]
\centering
\small
\setlength{\tabcolsep}{5pt}
\renewcommand{\arraystretch}{1.12}
\begin{tabular}{@{}>{\raggedright\arraybackslash}p{0.23\linewidth}>{\raggedright\arraybackslash}p{0.71\linewidth}@{}}
\toprule
Symbol & Meaning\\
\midrule
$x$, $\mu$ & Background coordinate and endpoint at which amplitudes are evaluated.\\[3pt]
$\h$, $g$, $n$ & Formal expansion parameter, genus (number of handles), and number of boundaries.\\[3pt]
$D=\h\partial_x$; $Q$, $P$ & Scaled derivative, scalar operator, and its string partner. Their orders are denoted by $(r,s)$.\\[3pt]
$u_j$, $u_{j,g}$ & Coefficient fields of $Q$ and their coefficients at order $\h^{2g}$.\\[3pt]
$p$; $E(p,x)$ & Coordinate on the classical curve and its energy projection.\\[3pt]
$\Delta=E_p$, $\alpha_i$ & Derivative of the energy projection and its finite zeros, the ramification points.\\[3pt]
$V=E_x|_p$, $v_i$ & Motion of the leading energy polynomial and branch-value velocities $v_i=V(\alpha_i,x)$.\\[3pt]
$\cD=\partial_x|_E$ & Derivative along the moving curve while holding the energy fixed.\\[3pt]
$y$; $B(p_1,p_2)$ & Disk function, with $\omega_{0,1}=y\,\dd E$, and the Bergman two-point differential.\\[3pt]
$\Rhat$, $\Rhat_{2g}$ & Normalized scalar resolvent and its coefficient at order $\h^{2g}$.\\[3pt]
$\Qhat_g$ & Positive-genus resolvent primitive, satisfying $\cD\Qhat_g=\Rhat_{2g}$.\\[3pt]
$\omega_{g,n}$, $W_{g,n}$ & Amplitude differential and its coefficient in curve coordinates: $\omega_{g,n}=W_{g,n}\prod_i\dd p_i$.\\[3pt]
$\cM$, $\mathsf B$ & Sheet projector and spectral connection used in the matrix description.\\
\bottomrule
\end{tabular}
\caption{Notation used in the main construction. The genus $g$ labels the surfaces counted by the expansion; the classical curves themselves have a rational parametrization.}
\label{tab:notation}
\end{table}

Table~\ref{tab:benchmarks} gives an overview of the results and the examples through which we develop them. Here $(r,s)$ lists the orders of $Q$ and $P$, respectively. EO denotes an independent Eynard--Orantin calculation from the endpoint energy, disk and Bergman differential. A full comparison means equality of the rational functions in all their displayed variables. This lets us distinguish comparisons with the literature, independent recursion calculations, and checks at selected values of the arguments.

\begin{table}[p]
\centering
\small
\setlength{\tabcolsep}{4pt}
\renewcommand{\arraystretch}{1.08}
\begin{tabular}{@{}>{\raggedright\arraybackslash}p{0.215\linewidth}>{\raggedright\arraybackslash}p{0.14\linewidth}>{\raggedright\arraybackslash}p{0.59\linewidth}@{}}
\toprule
Setting & Range & Result and comparison\\
\midrule
\multicolumn{3}{@{}l}{\textit{General construction}}\\[3pt]
Regular type $A_{r-1}$ & All $r\geq2$, all $g\geq1$ & At most one admissible primitive and one set of field corrections at each genus. Existence is conditional on compatibility; an admissible string solution supplies it. Theorem~\ref{thm:unique} and Proposition~\ref{prop:transfer}.\\[5pt]
Specified string system & General $g,n$ & Exact trace primitive \eqref{eq:traceprimitive} and cyclic-trace construction \eqref{eq:Wcycles}. Identification with EO additionally uses topological type \cite{BergereBorotEynard2015}.\\[5pt]
Generic third order & $g=1$ & Both moving-field corrections are explicit for general regular leading fields in \eqref{eq:genericB}--\eqref{eq:genericA}.\\[5pt]
Unitary $(r,r+1)$ & All $r\geq2$ & Closed $W_{0,3}$, $W_{1,1}$ and highest-field genus-one correction: \eqref{eq:Cheby03}, \eqref{eq:Cheby11}, \eqref{eq:Chebytop}.\\
\midrule
\multicolumn{3}{@{}l}{\textit{Explicit backgrounds and independent checks}}\\[3pt]
Pure gravity $(3,2)$ & Fields to $g=3$ & Painlev\'e~I through $\h^6$; $W_{1,1},W_{2,1}$ compared with the literature and EO; full matrix/EO $W_{1,2}$ \cite{BergereBorotEynard2015}. Genus-three checks use the string and trace identities.\\[5pt]
Ising $(3,4)$ & Fields to $g=3$, $a\equiv0$ & $W_{g,1}$ through $g=3$ compared with the literature, with disk normalization fixed; full matrix/EO $W_{1,2}$ \cite{BergereBorotEynard2015}. The stated reconstruction uses the zero-magnetic-field sector.\\[5pt]
$(4,5)$, $(5,6)$ & Fields to $g=2$ & String commutators through $\h^5$; full independent EO $W_{1,1},W_{2,1},W_{1,2}$. All homogeneous fields are allowed to vary.\\[5pt]
$(6,7)$, $(7,8)$ & Fields to $g=2$ & String and trace-primitive checks; full EO $W_{1,1},W_{2,1}$. Full forward-projection $W_{1,2}$ remains to be compared.\\[5pt]
Dual $(4,3)$, $(5,4)$, $(6,5)$ & Fields to $g=2$ & Complete dual operators, with the Ising pair compared with the literature; full EO $W_{1,1},W_{2,1},W_{1,2}$. Boundary functions differ between projections.\\[5pt]
Nonunitary $(3,5)$, $(5,3)$ & Fields to $g=2$ & Operators and both integrated string equations compared with the literature; full EO $W_{1,1},W_{2,1},W_{1,2}$, including complex ramification \cite{ChanIrieNiednerYeh2016}.\\[5pt]
Yang--Lee $(5,2)$, $(2,5)$ & Fields to $g=2$ & Fifth-/second-order pair and KdV string equation compared with the literature; full EO $W_{1,1},W_{2,1},W_{1,2}$ \cite{ChanIrieNiednerYeh2016}.\\[5pt]
Higher-Airy, $r=4,5$ & $W_{g,1}$ to $g=6$ & $r$-spin one-point coefficients compared with the literature; full $W_{0,3},W_{0,4},W_{1,2}$ intersection checks \cite{LiuVakilXu2011,EynardEtAl2023}. Tests the specified operator on the higher-Airy locus.\\
\bottomrule
\end{tabular}
\caption{Summary of the reconstruction, boundary construction and explicit comparisons. The finite genus bounds describe the stated homogeneous families. The general results in the upper part follow from the theorem and trace construction under their stated hypotheses.}
\label{tab:benchmarks}
\end{table}

Some further multi-boundary calculations are included in the appendices with a more limited set of checks. For Ising, the four-boundary result has eight exact matrix-cycle evaluations, while $W_{2,2}(0,p)$ has a separate projector extraction and local-pole checks. The corresponding full four- and two-variable comparisons remain open. Appendix~\ref{app:isingmulti} describes their scope, and Appendix~\ref{app:airydata} does the same for the higher-Airy coefficients beyond the fully compared cases in Table~\ref{tab:benchmarks}.

The paper is organized as follows. Section~\ref{sec:curve} develops the scalar reconstruction, and Section~\ref{sec:backgrounds} connects it to string equations and several boundaries. Section~\ref{sec:ising} works through pure gravity and Ising; Section~\ref{sec:unitary} extends the examples to higher-order, dual and nonunitary models. Section~\ref{sec:airy} studies higher-Airy limits and their connection to intersection theory. Section~\ref{sec:discussion} discusses what the reconstruction determines and the questions it leaves open. The appendices collect the proofs, detailed calculations, longer formulas and conventions for comparison with the literature.

\section{From the spectral curve to quantum fields}\label{sec:curve}
\subsection{The curve, its motion, and the disk}
Before constructing the higher-order resolvent, it helps to recall what the differential equation describes. In the KdV formulation the operator
\begin{equation}
 \mathcal H=-\h^2\partial_x^2+u(x)
\end{equation}
is an auxiliary quantum-mechanical Hamiltonian. The coordinate $x$ comes from the double-scaled orthogonal-polynomial description. It labels the background along which the calculation proceeds; boundary length and spectral energy are separate variables. Integrating an appropriate diagonal matrix element up to $x=\mu$ then gives the macroscopic-loop observable. The matrix-model string equation selects the potential on which this calculation is performed \cite{BanksDouglasSeibergShenker1990,AhmedJohnsonSaraswat2025,DiFrancescoGinspargZinnJustin1995}.

With the resolvent normalization used in that formulation, the Gel'fand--Dikii equation reads \cite{GelfandDikii1975}
\begin{equation}\label{eq:GD}
 4(u-E)\Rhat^2-2\h^2\Rhat\Rhat''+\h^2(\Rhat')^2=1.
\end{equation}
A prime on a field or resolvent in this equation denotes an $x$ derivative at fixed $E$. Write $u=u_0+\h^2u_2+O(\h^4)$, set $X=u_0(x)-E$, and choose the branch $\Rhat_0=-1/(2\sqrt X)$. We work locally where $u_0'\ne0$ and on a chosen branch of $\sqrt X$, away from $X=0$ when evaluating the expansion. In this Schr\"odinger example, $u_2$ denotes the coefficient of $\h^2$ in the potential. On the string solution with
\begin{equation*}
 u_2=-\frac1{12}\partial_x^2\log u_0',
\end{equation*}
the genus-one resolvent admits the primitive \cite{Johnson2024Virasoro,AhmedJohnsonSaraswat2025}
\begin{equation}\label{eq:KdVprimitive}
 \Qhat_1=\frac{u_0'}{32X^{5/2}}
       -\frac{u_0''}{48u_0'X^{3/2}},
 \qquad \partial_x\Qhat_1=\Rhat_2.
\end{equation}
Here $\Rhat_2$ is the coefficient of $\h^2$, a convention we keep throughout. To check the required potential correction explicitly, differentiating the displayed primitive gives
\begin{equation*}
 \partial_x\Qhat_1
 =-\frac{5(u_0')^2}{64X^{7/2}}
  +\frac{u_0''}{16X^{5/2}}
  -\frac{\partial_x^2\log u_0'}{48X^{3/2}},
\end{equation*}
whereas expanding \eqref{eq:GD} gives the same first two terms and $u_2/(4X^{3/2})$ as its last term. The logarithm is defined on a local branch; changing it by an additive constant leaves its derivatives unchanged. The useful feature of \eqref{eq:KdVprimitive} is that the integral has become an endpoint expression: knowing the potential and its derivatives there is enough to evaluate it. Our aim is to retain this feature when there are several coefficient fields to determine.

For the higher-order construction we choose a monic operator, with leading term $+D^r$. Its second-order member is therefore $D^2+u_0$. This differs in sign from the Schr\"odinger convention above, so we will fix the resolvent normalization directly from the scalar pairing. The corresponding recursion sign is given in Section~\ref{sec:EO}.

Consider \eqref{eq:introQ} and a leading WKB solution
\begin{equation}
 \psi(x,E;\h)\sim\exp\!\left(\frac{S_0(x,E)}{\h}+\cdots\right),
 \qquad p=\partial_xS_0.
\end{equation}
At leading order, $D\psi/\psi=p$. Replacing $D$ by $p$ and each field by its leading coefficient gives the polynomial
\begin{equation*}
 E(p,x):=p^r+\sum_{j=0}^{r-2}u_{j,0}(x)p^j.
\end{equation*}
This is the classical energy projection in \eqref{eq:introE}. The notation $E(p,x)$ emphasizes that its value depends on both the curve coordinate $p$ and the background coordinate $x$. The coefficients $u_{j,0}(x)$ are the specified leading fields, rather than their full quantum expansions. For a fixed $x$, the relation between $E$ and $p$ defines the classical curve, and varying $p$ moves along that curve. Varying $x$ instead changes the polynomial itself and hence moves us through a family of curves.

The same letter $E$ in $Q\psi=E\psi$ denotes a prescribed spectral value. To make the distinction explicit, temporarily call that fixed value $E_\star$. Solving $E(p,x)=E_\star$ gives the points of the curve above that energy. Generically there are $r$ roots, each defining a local sheet $p=p(x,E_\star)$. Thus $p$ is a parameter on the curve, while the function $E(p,x)$ tells us which energy labels that point. We continue to write simply $E$ for the fixed spectral value when no confusion can arise.

There are two notions of genus in this discussion. The classical curve is parametrized by $p$ and compactifies to $\mathbb P^1$, so its geometric genus is zero. The integer $g$ in $\omega_{g,n}$ counts handles of the surfaces summed over in the perturbative expansion. Thus a genus-zero spectral curve can encode, for example, a genus-two boundary amplitude.

To locate the points where sheets meet, first differentiate the energy polynomial with respect to $p$ while holding $x$ fixed. We use a subscript $p$ for this partial derivative and denote the resulting polynomial by $\Delta$:
\begin{equation}\label{eq:Deltadef}
 \begin{aligned}
 \Delta(p,x)&:=E_p(p,x):=\left.\frac{\partial E(p,x)}{\partial p}\right|_x
       =rp^{r-1}+\sum_{j=1}^{r-2}j\,u_{j,0}(x)p^{j-1},\\
 \Delta(\alpha_i(x),x)&=0.
 \end{aligned}
\end{equation}
The first line defines one function, $\Delta(p,x)$. The second evaluates that same function at $p=\alpha_i(x)$ and requires it to vanish; it is the equation that determines the ramification coordinates $\alpha_i(x)$. The index $i$ labels its roots. Likewise, $\Delta_p=\partial_p\Delta=E_{pp}$ denotes a second derivative of the energy polynomial at fixed $x$.

At a simple root, $\Delta_p(\alpha_i,x)\ne0$, and the Taylor expansion is
\begin{equation*}
 E(p,x)-E(\alpha_i,x)
 =\frac12\Delta_p(\alpha_i,x)(p-\alpha_i)^2
       +O\bigl((p-\alpha_i)^3\bigr).
\end{equation*}
The linear term vanishes because $\Delta(\alpha_i,x)=0$. The two signs of the resulting local square root describe the two sheets that meet there. A generic order-$r$ operator has $r-1$ such finite points. Increasing the order gives more branch points, while making them collide requires a special choice of coefficients. For example,
\begin{equation}
 E=p^4-6p^2+8p,
 \qquad \Delta=4(p-1)^2(p+2)
\end{equation}
has a partial collision. The pure $r$-Airy projection $E=p^r$ has a maximally degenerate finite ramification point. These loci require separate treatment from the simple-root calculation.

In spectral recursion this is the distinction between ordinary simple-ramification recursion and its higher-ramification extension \cite{EynardOrantin2007,BouchardEynard2013}.

The leading fields depend on $x$, giving a family of curves as the background changes. Define
\begin{equation}\label{eq:Vdef}
 V(p,x)=\left.\partial_xE(p,x)\right|_p
       =\sum_{j=0}^{r-2}u_{j,0}'(x)p^j.
\end{equation}
Unlike $E_p$, this derivative keeps the curve coordinate $p$ fixed and differentiates the leading fields. It measures the change in energy at the same value of $p$ as the background changes. Holding the \emph{energy} fixed is a different operation. On a chosen sheet, differentiate $E(p(x,E_\star),x)=E_\star$ by the chain rule:
\begin{equation*}
 0=\left.\frac{\partial E(p,x)}{\partial x}\right|_p
   +\left.\frac{\partial E(p,x)}{\partial p}\right|_x
       \left.\frac{\partial p}{\partial x}\right|_{E_\star}
  =V+\Delta\left.\frac{\partial p}{\partial x}\right|_{E_\star}.
\end{equation*}
Away from $\Delta=0$, this gives $\partial_xp|_E=-V/\Delta$. A function of $p$ and the background fields changes both explicitly with $x$ and through this induced change in $p$. Its total derivative is therefore
\begin{equation}\label{eq:totalderivative}
 \cD:=\left.\partial_x\right|_E
 =\left.\partial_x\right|_p-\frac{V}{\Delta}\partial_p.
\end{equation}
The first term differentiates every field and its derivatives. For example, if $f$ depends on the jet variables $u_{j,0}^{(m)}$, then this term contains $\sum_{j,m}u_{j,0}^{(m+1)}\partial f/\partial u_{j,0}^{(m)}$. A jet means only the value of a field together with finitely many of its derivatives.

The energies of the ramification points also move. Their velocities are
\begin{equation}\label{eq:velocities}
 v_i=\frac{\dd}{\dd x}E(\alpha_i(x),x)=V(\alpha_i,x).
\end{equation}
The term involving $\alpha_i'$ vanishes because $E_p(\alpha_i)=0$. Thus $v_i$ measures the motion of a branch \emph{value} in the energy plane; $\alpha_i'$ measures the motion of its ramification coordinate. The reconstruction below works locally where
\begin{equation}\label{eq:regular}
 \Delta_p(\alpha_i,x)\ne0,\qquad v_i\ne0
 \quad\hbox{for every finite ramification point.}
\end{equation}
The first condition keeps each branch point simple; the second ensures that its energy value moves with the background. They allow real or complex roots, and several branch points may have the same energy value. This distinction will allow the same calculation to cover both the unitary and nonunitary examples.

At the endpoint $x=\mu$, the energy projection is only one part of the recursion data. We use
\begin{equation}\label{eq:spectraldata}
 \omega_{0,1}(p)=y(p)\,\dd E(p,\mu),
 \qquad B(p_1,p_2)=\omega_{0,2}(p_1,p_2)
 =\frac{\dd p_1\dd p_2}{(p_1-p_2)^2}.
\end{equation}
The function $y$ is the disk amplitude, determined by the background and the selected normalization. The pair of functions $p\mapsto(E(p,\mu),y(p))$, together with $B$, is the spectral curve data used for the amplitudes.

For positive genus, the scalar primitive and the one-boundary differential are related by
\begin{equation}\label{eq:oneboundary}
 W_{g,1}(p)=\Delta(p,\mu)\Qhat_g(p,\mu),
 \qquad \omega_{g,1}=W_{g,1}(p)\dd p.
\end{equation}
More generally we write
\begin{equation}\label{eq:omegaW}
 \omega_{g,n}(p_1,\ldots,p_n)
 =W_{g,n}(p_1,\ldots,p_n)\prod_{i=1}^n\dd p_i.
\end{equation}
The displayed $W$'s are therefore coefficients of differentials in curve coordinates. A thermal amplitude or a geometric volume is obtained by a further, model-dependent transformation. Keeping the differential explicit makes those later changes of variables easier to track and helps specify the geometric interpretation appropriate to each model.

For ordinary JT gravity, the geometric interpretation is fixed by the matrix-integral description and trumpet gluing \cite{SaadShenkerStanford2019}, and connects with Mirzakhani's volume recursion and its matrix-model formulation \cite{Mirzakhani2007,EynardOrantin2007WP}. Supersymmetric volume recursions require different geometric data; recent spectral-curve and ODE treatments make that distinction explicit \cite{Johnson2026SuperRecursion,AhmedJohnsonSaraswat2025}.

At genus zero, $\cD\Qhat_0=1/\Delta$ determines the disk primitive up to a function of $E$. One must perform this integration at fixed energy before evaluating $x=\mu$. Locally where $V\ne0$, equation~\eqref{eq:totalderivative} lets us use $p$ as a coordinate along that fixed-energy trajectory and gives
\begin{equation}\label{eq:diskintegral}
 \Qhat_0=-\int\frac{\dd p}{V},
 \qquad y(p)=\Qhat_0(p,\mu).
\end{equation}
The integral follows a trajectory of fixed energy, along which the fields inside $V$ also evolve. This is why the endpoint function $y(p)$ must be found after carrying out the integral. At a zero of $V$, this change of integration variable is unavailable, but the original equation $\cD\Qhat_0=1/\Delta$ remains the defining relation wherever $\Delta\ne0$. An expression obtained on a patch with $V\ne0$ may be continued through a zero of $V$ when the resulting primitive is regular. For example, the Ising background has $V=-p/U^2$, yet its disk is regular at $p=0$. One may add a common function $C(E)$: it cancels between the two locally conjugate points in spectral recursion and leaves the stable amplitudes unchanged. The remaining part of the disk retains the background dependence that distinguishes models with the same endpoint energy polynomial.

\subsection{The normalized scalar resolvent}\label{sec:resolvent}
We next need a normalization of the resolvent that works at any order. The first step is to fix where the derivatives sit in the operator. Since $Du=uD+\h u'$, moving $D$ past a field produces an extra term. At third order, our symmetric convention \eqref{eq:introQ} gives
\begin{equation}\label{eq:Q3}
 Q=D^3+\frac12\{b,D\}+a
   =D^3+bD+a+\frac\h2b',
 \qquad a=u_0,\quad b=u_1.
\end{equation}
The last term is completely determined by $b$. It belongs to the ordering prescription, even though it carries an odd power of $\h$. Leaving it out while keeping both fields even would already introduce an unwanted term $-3\h b'p/\Delta^3$ in the formal product. This small example explains why ordering must be settled before comparing genus coefficients.

For computation, write the same operator in left-ordered form,
\begin{equation}\label{eq:leftcoeff}
 \begin{aligned}
 Q&=\sum_{j=0}^r c_jD^j,\qquad c_r=1,\qquad c_{r-1}=0,\\
 c_j&=u_j+\frac12\sum_{k=j+1}^{r-2}\binom{k}{j}
          \h^{k-j}\partial_x^{k-j}u_k,
          \qquad 0\leq j\leq r-2.
 \end{aligned}
\end{equation}
Thus the left-ordered coefficients $c_j$ are determined by the symmetric fields $u_j$ and their derivatives. At fourth order, for example, $\{u_2,D^2\}/2$ contributes $u_2D^2$, $\h u_2'D$ and $\h^2u_2''/2$. All three terms describe the same field in the chosen ordering, and all are retained below.

The formal adjoint of the left-ordered operator is
\begin{equation}
 Q^*=\sum_{k=0}^r(-D)^k\circ c_k.
\end{equation}
Let $Q\psi=E\psi$ and $Q^*\chi=E'\chi$. The higher-order version of Wronskian conservation is the Lagrange identity
\begin{align}
 \mathcal B(\chi,\psi)&=
 \h\sum_{k=1}^r\sum_{j=0}^{k-1}(-1)^j
       D^j(c_k\chi)D^{k-1-j}\psi,\label{eq:concomitant}\\
 \partial_x\mathcal B&=(E-E')\psi\chi.\label{eq:Lagrange}
\end{align}
This pairing is the bilinear concomitant. At equal energies its derivative vanishes, so it plays the same normalizing role as the conserved Wronskian of a second-order equation. We pair solutions with opposite leading WKB momenta and choose $\mathcal B=+\h$.

The solutions here are formal series on a chosen sheet. That is enough for the perturbative amplitudes we study. To obtain an analytic Green function one would also have to select actual solutions and their contours or boundary conditions. The concomitant and its differential-system interpretation are standard ingredients of the determinantal construction \cite{BergereEynard2009,BergereBorotEynard2015}.

Logarithmic derivatives express the wavefunction equations in terms of scalar functions that can be found recursively. Introduce
\begin{equation}
 w=\h\frac{\partial_x\psi}{\psi},\qquad
 \widetilde w=\h\frac{\partial_x\chi}{\chi}.
\end{equation}
The polynomials
\begin{equation}\label{eq:Bell}
 \cP_0(w)=1,\qquad
 \cP_{m+1}(w)=w\cP_m(w)+\h\cD\cP_m(w)
\end{equation}
satisfy $D^m\psi=\cP_m(w)\psi$. Each further derivative is therefore obtained by one multiplication and one differentiation. The first terms are $1,w,w^2+\h w'$ and $w^3+3\h ww'+\h^2w''$. We use $\cP_m$ for these polynomials; ordinary $P$ will denote the string partner.

Leibniz's rule similarly gives
\begin{equation}\label{eq:Pi}
 \Pi_j[c_k]=\sum_{i=0}^j\binom ji\h^i
   (\partial_x^i c_k)\cP_{j-i}(\widetilde w),
 \qquad D^j(c_k\chi)=\Pi_j[c_k]\chi.
\end{equation}
Factoring $\psi\chi$ out of \eqref{eq:concomitant} leaves a finite expression:
\begin{equation}\label{eq:Sigma}
 \Sigma=\sum_{k=1}^r\sum_{j=0}^{k-1}(-1)^j
               \Pi_j[c_k]\cP_{k-1-j}(w),
 \qquad \Rhat:=\psi\chi=\frac1\Sigma.
\end{equation}
The finite sum $\Sigma$ gives us the higher-order counterpart of the quadratic identity \eqref{eq:GD}. Once the logarithmic derivatives are known, the resolvent follows by series inversion. At leading order $w_0=p$ and $\widetilde w_0=-p$; the $k$ terms associated with $c_k$ then sum to $kc_kp^{k-1}$. Thus the normalization has a simple answer at every scalar order:
\begin{equation}\label{eq:Rleading}
 \Sigma_0=\Delta,\qquad \Rhat_0=\frac1\Delta.
\end{equation}

The symmetric operator with even fields obeys $Q^*(\h)=Q(-\h)$. With the paired branch convention, this implies
\begin{equation}\label{eq:parity}
 \widetilde w(\h)=-w(-\h),\qquad
 \Rhat(\h)=\sum_{g\geq0}\h^{2g}\Rhat_{2g}.
\end{equation}
This even expansion belongs to the normalized scalar pairing. The logarithmic derivative $w$ still has odd coefficients, as does the matrix projector used later for several boundaries. Those coefficients must be kept when building either object.

The scalar equation becomes $\sum_jc_j\cP_j(w)=E$. Expand $w=p+\sum_{m\geq1}\h^mw_m$ and set $w_{<m}=p+\sum_{k=1}^{m-1}\h^kw_k$. The coefficient multiplying the new $w_m$ is always $\Delta$, so
\begin{equation}\label{eq:WKBrec}
 w_m=-\frac1\Delta[\h^m]
       \left(\sum_{j=0}^r c_j\cP_j(w_{<m})-E\right).
\end{equation}
Square brackets extract the indicated power of $\h$. To find $w_m$, we substitute everything known at lower orders, take the coefficient of $\h^m$, and divide by $-\Delta$. The coefficient fields may remain symbolic if they are still to be reconstructed. This algebraic step determines each $w_m$ directly from the earlier coefficients.

After computing $\Sigma=\sum_m\h^m\Sigma_m$, its inverse is obtained recursively by
\begin{equation}\label{eq:Rrec}
 \Rhat_m=-\frac1\Delta\sum_{k=1}^m\Sigma_k\Rhat_{m-k},
 \qquad m\geq1.
\end{equation}
These recurrences give a practical sequence: put the operator in left-ordered form, find the logarithmic derivatives, evaluate the concomitant, and invert its series. Each step uses differentiation and finite algebra at the order being calculated. The same sequence applies at every scalar order, with a larger concomitant sum as $r$ increases.

At third order we can write these steps explicitly. Equation~\eqref{eq:Q3} gives
\begin{equation}\label{eq:riccati3}
 w^3+3\h ww'+\h^2w''+bw+a+\frac\h2b'=E,
\end{equation}
and its concomitant is
\begin{equation}\label{eq:Sigma3}
 \Sigma=w^2+\widetilde w^2-w\widetilde w
       +\h(w'+\widetilde w')+b.
\end{equation}
Here the leading denominator is $\Delta=3p^2+b_0$. We have reduced the two-field problem to scalar rational functions containing the unknown corrections to $a$ and $b$. The analytic conditions on the primitive will now supply the missing equations.

\subsection{Recovering the quantum fields}\label{sec:reconstruction}
Suppose we have completed the calculation through genus $g-1$. At genus $g$, both the next field corrections and the primitive are unknown. We seek a rational function $\Qhat_g$ whose derivative gives $\Rhat_{2g}$, whose finite poles lie only at zeros of $\Delta$, and whose differential $\Qhat_g\Delta\dd p$ is regular at infinity. Let $\mathbb K$ denote the field generated by the leading fields, the known lower-genus corrections, their finitely many $x$ derivatives and the fixed model parameters. The allowed primitive then belongs to $\mathbb K(p)$.

Put $d=\deg\Delta=r-1$. Repeated polynomial division gives the useful normal form
\begin{equation}\label{eq:slots}
 \Qhat_g=\sum_{k=2}^{K_g^{\max}}\frac{A_k(p)}{\Delta^k},
 \qquad \deg A_k<d,\qquad \deg A_2\leq d-2.
\end{equation}
The normal form makes the condition at infinity explicit. A polynomial part or a $\Delta^{-1}$ term would make the differential singular there, so both are absent. For $\Delta^{-2}$, the highest possible numerator coefficient must also vanish, giving the last inequality; at $r=2$ this sets $A_2=0$. We are therefore imposing a boundary condition on the amplitude differential before attempting to determine any field.

We call each denominator power a pole slot. It is a way of organizing the calculation with polynomials, while the actual poles lie at the roots of $\Delta$. Working directly with the polynomial keeps the calculation equally accessible when the roots are real, complex or inconvenient to express in radicals.

Suppose the fields have already been found below genus $g$. Collect the new corrections into the polynomial
\begin{equation}\label{eq:Ffield}
 F_g(p)=\sum_{j=0}^{r-2}u_{j,g}(x)p^j.
\end{equation}
The polynomial $F_g$ collects the new fields we want to find as its coefficients. Everything already known is collected in $S_g$, the coefficient of $\h^{2g}$ in the scalar recurrence with the new fields set to zero and the lower-genus corrections retained. The full coefficient is then
\begin{equation}\label{eq:fieldvariation}
 \boxed{\displaystyle
 \Rhat_{2g}=S_g+\frac{F_g\Delta_p-(\partial_pF_g)\Delta}{\Delta^3}.}
\end{equation}
It is useful to derive this formula in two steps: first change the fields at fixed $x$ and fixed spectral energy, then identify the power of $\h$ at which that change contributes. Introduce an auxiliary parameter $\epsilon$ and deform the energy polynomial by
\begin{equation*}
 E_\epsilon(p,x)=E(p,x)+\epsilon F_g(p,x).
\end{equation*}
For this variation, $x$ is held fixed. We also hold the spectral value $E_\star$ fixed, so the chosen root must change from $p$ to $p_\epsilon$:
\begin{equation*}
 E_\epsilon(p_\epsilon,x)=E(p,x)=E_\star,
 \qquad p_\epsilon=p+\epsilon\eta_g+O(\epsilon^2).
\end{equation*}
Here $\eta_g$ is the coefficient of the root displacement under the field variation. It is not the background derivative $\cD p$. Expanding the first equality gives
\begin{equation*}
 0=\epsilon\bigl(\Delta\eta_g+F_g\bigr)+O(\epsilon^2),
 \qquad \eta_g=-\frac{F_g}{\Delta}.
\end{equation*}
This is the first effect of the unknown fields: even at the same energy, they change which point of the deformed curve lies on our chosen sheet.

The leading resolvent is the reciprocal of the derivative of the energy polynomial, evaluated at that root. Its denominator changes in two ways. There is the explicit derivative of the added polynomial, and there is the change from evaluating the old derivative at $p_\epsilon$ rather than $p$:
\begin{align*}
 \Delta_\epsilon(p_\epsilon,x)
 &:=\left.\partial_qE_\epsilon(q,x)\right|_{q=p_\epsilon}\\
 &=\Delta+\epsilon\bigl(\partial_pF_g+\eta_g\Delta_p\bigr)
       +O(\epsilon^2).
\end{align*}
The variable $q$ here is only a temporary argument for the differentiation. All quantities on the second line are evaluated at the undeformed point $(p,x)$. Taking the reciprocal and inserting $\eta_g=-F_g/\Delta$ gives
\begin{align*}
 \frac1{\Delta_\epsilon(p_\epsilon,x)}
 &=\frac1\Delta
   -\epsilon\frac{\partial_pF_g+\eta_g\Delta_p}{\Delta^2}
     +O(\epsilon^2)\\
 &=\frac1\Delta
   +\epsilon\frac{F_g\Delta_p-(\partial_pF_g)\Delta}{\Delta^3}
     +O(\epsilon^2).
\end{align*}
The two numerator terms in \eqref{eq:fieldvariation} thus have separate origins: $-(\partial_pF_g)\Delta$ comes from changing the polynomial's derivative, while $F_g\Delta_p$ comes from moving its root. Holding $p$ fixed during this variation would miss the latter contribution.

We now set $\epsilon=\h^{2g}$ to insert the actual genus-$g$ fields. Before expanding the coefficient fields themselves, the normalized scalar resolvent has the form of a leading, derivative-free term plus explicit dispersive corrections,
\begin{equation*}
 \Rhat[u;\h]=\Rhat^{\rm cl}[u]
           +\h^2\mathcal R^{(2)}[u]
           +\h^4\mathcal R^{(4)}[u]+\cdots.
\end{equation*}
Here $u$ denotes the collection of coefficient fields; the functionals $\mathcal R^{(2m)}$ also involve their $x$ derivatives. Inserting a new field of order $\h^{2g}$ into a term already multiplied by $\h^{2m}$, with $m\geq1$, changes it only at order $\h^{2g+2m}$ or higher. Products with an already known positive-genus field are also of higher order, as are terms quadratic in the new fields. Consequently their contribution at order $\h^{2g}$ is only the first variation of $\Rhat^{\rm cl}=1/\Delta$ about the leading fields. All contributions involving known lower-genus fields are precisely the source $S_g$. This establishes \eqref{eq:fieldvariation} without using a string equation.

For a concrete check, take $r=3$ and write the new genus-one fields as $A=u_{0,1}$ and $B=u_{1,1}$. Then $F_1=A+Bp$ and $\Delta=3p^2+b_0$, so the general expression becomes
\begin{equation*}
 \Rhat_2-S_1=
 \frac{6pA+(3p^2-b_0)B}{(3p^2+b_0)^3}.
\end{equation*}
Both field corrections enter linearly, and no derivatives of $A$ or $B$ appear at this order.

This is the separation we need. The unknown fields affect poles of order at most three. The stronger poles are already fixed by the lower-genus data. We can match those stronger poles to find a candidate primitive before solving for the fields themselves.

For a single term in \eqref{eq:slots}, the fixed-energy derivative is
\begin{equation}\label{eq:slotderivative}
 \cD\left(\frac A{\Delta^k}\right)
 =\frac{A_x}{\Delta^k}
 -\frac{VA_p+kAV_p}{\Delta^{k+1}}
 +\frac{kAV\Delta_p}{\Delta^{k+2}}.
\end{equation}
Here $A_x$ is taken at fixed $p$. The last term is the only one that reaches denominator power $k+2$. At this power we keep the numerator modulo $\Delta$, so the relevant operation is multiplication by $kV\Delta_p$ in the finite algebra $\mathbb K[p]/(\Delta)$. Its determinant is
\begin{equation}\label{eq:determinant}
 \det\mathscr L_k=k^{r-1}\prod_{i=1}^{r-1}v_i\Delta_p(\alpha_i),
 \qquad \mathscr L_k[A]=[kV\Delta_pA].
\end{equation}
The two regularity conditions \eqref{eq:regular} make every factor nonzero. We can therefore solve for the numerator of the strongest pole and then move to the next one. In practice, the extended Euclidean algorithm finds the required inverse directly from the polynomial coefficients.

The direction of this construction is already useful in KdV: matching the strongest inverse powers first gives the algebraic, descending primitive calculation of Ahmed, Johnson and Saraswat \cite{AhmedJohnsonSaraswat2025}. Here the single-field matching is replaced by polynomial multiplication modulo $\Delta$, so the same logic can constrain all $r-1$ coefficient fields.

The scalar recurrence has pole order at most $6g+1$ in $\Rhat_{2g}$ at simple ramification. Equation~\eqref{eq:slotderivative} then bounds an admissible primitive by $K_g^{\max}=6g-1$. Starting with $S_g$, write $H_{k+2}$ for the numerator of the current strongest slot. The next primitive numerator is
\begin{equation}\label{eq:descent}
 A_k=\remd_{\Delta}\!\left(
       \frac{(V\Delta_p)^{-1}H_{k+2}}{k}\right),
 \qquad k=6g-1,\ldots,2.
\end{equation}
At each step we add $A_k/\Delta^k$ to the primitive and subtract its full derivative from the remainder. The next step starts with the strongest pole still present. Thus $H_{k+2}$ always refers to the reduced numerator after the earlier subtractions, and every inverse in \eqref{eq:descent} is taken modulo $\Delta$.

Once all poles of order four and higher have been removed, define
\begin{equation}\label{eq:fieldrecovery}
 K_g=\Delta^3(\cD\Qhat_g-S_g),
 \qquad F_g=\remd_{\Delta}\bigl((\Delta_p)^{-1}K_g\bigr).
\end{equation}
The calculation is accepted only if the full residual identities hold:
\begin{equation}\label{eq:obstructions}
 K_g=F_g\Delta_p-(\partial_pF_g)\Delta,
 \qquad [p^{r-2}]A_2=0.
\end{equation}
The first identity checks the entire remaining rational expression, beyond the congruence used to recover $F_g$. The second imposes the condition at infinity. If either fails, the proposed leading family has an obstruction at this genus in the chosen primitive class.

To see why the answer is unique, suppose there were two allowed primitives and subtract them. If their difference were nonzero, it would have a highest occupied slot $k\geq2$. Its derivative would then have a nonzero pole of order $k+2\geq4$, by \eqref{eq:determinant}. The difference of the two field contributions has pole order at most three. Equality therefore forces the primitive difference to vanish. The degree bound on $F_g$ then forces the fields to coincide as well. Appendix~\ref{app:proofs} gives the precise statement. The result holds at every scalar order and genus within the stated analytic class, with existence decided by the compatibility conditions \eqref{eq:obstructions}.

The generic third-order calculation makes the mechanism explicit. In this subsection $a(x),b(x)$ denote the leading fields of \eqref{eq:Q3}, with corrections $a+\h^2A$ and $b+\h^2B$. The ramification data are
\begin{equation}
 p_\pm=\pm\sqrt{-b/3},\quad v_\pm=a'+b'p_\pm,
 \quad v_+v_-=a'^2+\frac b3(b')^2.
\end{equation}
On the domain $b\ne0$ and $v_+v_-\ne0$, the pole calculation leaves three algebraic conditions on the two corrections. They have coefficient and augmented rank two and give
\begin{align}
 B&=\frac18\partial_x^2\log\!\left(3(a')^2+b(b')^2\right),\label{eq:genericB}\\
 A&=\frac18\partial_x\left[
       \sqrt{-b/3}\,\partial_x\log\!\left(\frac{v_+}{v_-}\right)
                        \right].\label{eq:genericA}
\end{align}
These expressions may be evaluated on any local branches of the logarithms, since additive constants disappear on differentiation. Substitution into the scalar recurrence gives the required rational derivative identity. At this genus the compatibility conditions hold for general regular leading functions $a$ and $b$, so the calculation determines both corrections throughout that family.

The top-field result may be written as
\begin{equation}\label{eq:Floc}
 B=3\partial_x^2F_1^{\rm loc},\qquad
 F_1^{\rm loc}=\frac1{24}(\log v_++\log v_-).
\end{equation}
The motion of the branch values controls this correction. The superscript reminds us that $F_1^{\rm loc}$ is a potential for a field correction. Its relation to the coefficient of $\log Z$ depends on the field-to-tau-function convention; Ising provides an explicit sign check below.

Genus-one corrections to integrable hierarchies also have a broader Frobenius-manifold and bihamiltonian formulation \cite{Dubrovin1996,DubrovinZhang1998}. Our formulas above are stated in symmetric operator coefficients, so a comparison with flat-coordinate deformation formulas requires the corresponding change of variables.

The compact third-order answer makes it tempting to look for every higher-order correction by interpolating branch velocities. This reproduces \eqref{eq:genericB}--\eqref{eq:genericA}. At fourth order, the lowest field of a moving family acquires an additional contribution, calculated in Appendix~\ref{app:diagnostics}. For higher scalar order we therefore return to the remainder construction and determine all field corrections independently.

\section{From string pairs to several boundaries}\label{sec:backgrounds}
\subsection{Choosing a background}
So far we have treated the leading family as given. To choose a physical model we need additional information, supplied by the integrable hierarchy \cite{Douglas1990,DrinfeldSokolov1985}. This is also the route by which we will compare the reconstructed fields with string equations in the literature. Begin with the formal root
\begin{equation}\label{eq:root}
 \mathcal L=Q^{1/r}=D+\sum_{j\geq1}\gamma_jD^{-j},
 \qquad \mathcal L^r=Q.
\end{equation}
Negative powers of $D$ here are formal symbols. Their multiplication is defined by the generalized Leibniz rule
\begin{equation}\label{eq:Leibniz}
 D^k f=\sum_{j\geq0}\binom kj\h^j f^{(j)}D^{k-j},
 \qquad k\in\mathbb Z.
\end{equation}
For $k\geq0$ the sum terminates; for $k<0$ it is a formal series. At any specified coefficient and order in $\h$, only finitely many terms contribute.

The commuting hierarchy flows are
\begin{equation}\label{eq:hierarchy}
 B_m=(Q^{m/r})_+,
 \qquad \h\partial_{t_m}Q=[B_m,Q],\qquad r\nmid m.
\end{equation}
The subscript $+$ retains the nonnegative powers of $D$. The times $t_m$ are deformation parameters, with $t_1$ identified with the spatial flow. The nontrivial congruence classes of $m$ modulo $r$ give $r-1$ primary sectors \cite{DrinfeldSokolov1985,BertolaDubrovinYang2016}. The energy projection also has its own set of $r$ sheet labels; we will explain their relation to the time sectors below.

For minimal-string backgrounds we specify a second differential operator $P$ and impose the Douglas string equation \cite{Douglas1990}
\begin{equation}\label{eq:stringpair}
 [P,Q]=\kappa\h,\qquad \kappa\ne0.
\end{equation}
The pure-flow examples below use $P=(Q^{s/r})_+$, with $r,s\geq2$ coprime. We list the order of $Q$ first and call this an ordered pair $(r,s)$. The matter central charge is
\begin{equation}\label{eq:centralcharge}
 c=1-\frac{6(r-s)^2}{rs}.
\end{equation}
The orders label the model. To calculate its amplitudes we also choose a hierarchy solution and a boundary projection \cite{DiFrancescoGinspargZinnJustin1995,SeibergShih2004,BelavinDubrovinMukhametzhanov2014,ChanIrieNiednerYeh2016}. The partner of order $s$ and the disk amplitude supply essential data alongside the degree-$r$ endpoint polynomial.

One consequence is already visible in Boussinesq. At $r=3$, the second flow has $B_2=D^2+2b/3$ and gives
\begin{equation}\label{eq:Boussinesq}
 \partial_{t_2}b=2a',\qquad
 \partial_{t_2}a=-\frac23bb'-\frac{\h^2}{6}b'''.
\end{equation}
This flow generally changes $a$ even when the chosen background has $a=0$. Adding a boundary therefore explores field variations beyond this particular background slice. We will keep this distinction in mind when constructing the Ising multi-boundary amplitudes.

For the examples, a common scaling form makes the background motion easy to follow. Assign weights $[p]=1$ and $[U]=2$, and define
\begin{equation}\label{eq:homvars}
 H=r+s,\qquad \delta=\frac{r+s-1}{2},\qquad
 U=x^{1/\delta},\qquad z=\frac p{\sqrt U}.
\end{equation}
Then $[x]=r+s-1$ and $[\h]=H$, so $D=\h\partial_x$ has weight one. We work on a chosen branch with $x\ne0$ and set the endpoint to $x=1$ for the displayed amplitudes. At that endpoint $U=1$ and $z=p$.

Write the leading pair as
\begin{equation}\label{eq:hompair}
 E(p,x)=U^{r/2}e(z),\qquad
 P_0(p,x)=U^{s/2}\pi(z),\qquad
 \pi=(e^{s/r})_+.
\end{equation}
Here $e$ and $\pi$ are monic polynomials of degrees $r$ and $s$; in the last expression the polynomial part is taken at $z=\infty$. The leading string equation is a constant Poisson bracket, with convention
\begin{equation}\label{eq:bracket}
 \{P_0,E\}=P_{0,p}E_x-P_{0,x}E_p
 =\frac{re\pi'-s\pi e'}{r+s-1}=\kappa.
\end{equation}
Primes on $e$ and $\pi$ mean $z$ derivatives. The weight assignment removes the overall $x$ dependence. The bracket condition then requires the remaining expression to be constant in $z$. Checking this bracket selects the leading homogeneous family before we start the quantum reconstruction.

The fixed-energy derivative now becomes a simple operation on rational functions:
\begin{equation}\label{eq:homderivative}
 \cD\bigl(U^\alpha f(z)\bigr)
 =U^{\alpha-\delta}\left[
       \frac\alpha\delta f(z)
       -\frac{re(z)}{2\delta e'(z)}f'(z)\right].
\end{equation}
Substituting \eqref{eq:bracket} shows that
\begin{equation}\label{eq:homdisk}
 \Qhat_0=-\frac{U^{s/2}\pi(z)}\kappa,
 \qquad y(z)=-\frac{\pi(z)}\kappa
 \quad\hbox{at }x=1.
\end{equation}
This fixes the disk normalization from the leading string pair. With both $E$ and $y$ specified, the ODE calculation and the later spectral-recursion comparison refer to the same boundary observable.

At higher genus, homogeneity separates the known dependence on $x$ from the unknown numerical coefficients. Since $u_jD^j$ has the same weight $r$ as $D^r$, the field $u_j$ has weight $r-j$. A term proportional to $\h^{2g}$ must therefore contain $U^{(r-j)/2-Hg}$, and we write
\begin{equation}\label{eq:homfields}
 u_j(x;\h)=\sum_{g\geq0}\beta_{j,g}\h^{2g}
                  U^{(r-j)/2-Hg},\qquad
 \beta_{j,0}=[z^j]e(z).
\end{equation}
The coefficient $\beta_{j,g}$ is a number independent of $x$ and $z$ once the homogeneous background and its fixed parameters have been chosen. Its first index tells us which field it corrects, and its second index is the genus. The leading numbers $\beta_{j,0}$ are the coefficients of the specified polynomial $e(z)$. The numbers $\beta_{j,g}$ with $g\geq1$ are what the reconstruction must determine. Thus the unknown function $u_{j,g}(x)$ has become an unknown constant multiplying a prescribed power of $U(x)$.

For example, in the Ising family $(r,s)=(3,4)$, one has $H=7$, $U=x^{1/3}$ and $b=u_1$. Its first two terms are
\begin{equation*}
 b(x;\h)=-3U+\beta_{1,1}\h^2U^{-6}+O(\h^4).
\end{equation*}
The power $U^{-6}=x^{-2}$ is fixed by scaling; the number $\beta_{1,1}$ is not. At each genus we solve for every allowed $\beta_{j,g}$, except where a sector restriction is explicitly imposed, as in the higher-genus Ising calculation. A coefficient that vanishes at leading order is still allowed to receive a quantum correction.

The primitive has a separate set of unknown coefficients. Its known overall scale is $U^{s/2-Hg}$, and its allowed poles occur at the zeros of $e'(z)$. Bringing the finite pole-slot expansion to a common denominator gives
\begin{equation}\label{eq:homprimitive}
 \Qhat_g=U^{s/2-Hg}\frac{N_g(z)}{e'(z)^{6g-1}},
 \qquad \deg N_g\leq(r-1)(6g-2)-2.
\end{equation}
Here $N_g(z)$ is an \emph{unknown polynomial}; the quotient in \eqref{eq:homprimitive} is the rational function. To state its unknowns explicitly, put
\begin{equation*}
 d_g^{\max}=(r-1)(6g-2)-2,\qquad
 N_g(z)=\sum_{\ell=0}^{d_g^{\max}}n_{g,\ell}z^\ell.
\end{equation*}
The numbers $n_{g,\ell}$ determine the primitive, whereas the numbers $\beta_{j,g}$ determine the fields of the operator. They are two different sets of coefficients, to be solved for together. The degree bound follows because the differential, at fixed $x$, is proportional to
$N_g(z)e'(z)^{-(6g-2)}\dd z$. If $N_g$ has degree $m$, its coefficient grows as $z^{m-(r-1)(6g-2)}$. Regularity at infinity requires this power to be at most $-2$, since $\dd z=-t^{-2}\dd t$ in the coordinate $t=1/z$. This gives exactly the bound in \eqref{eq:homprimitive}.

We can now make the coefficient-matching problem explicit. Remove the overall scale from the primitive and from the known source by defining
\begin{equation*}
 q_g(z)=\frac{N_g(z)}{e'(z)^{6g-1}},\qquad
 S_g(p,x)=U^{-(r-1)/2-Hg}s_g(z).
\end{equation*}
The function $s_g(z)$ is known from the scalar recurrence at this stage. Similarly, let $f_g(z)$ collect the unknown field coefficients:
\begin{equation*}
 f_g(z)=\sum_{j=0}^{r-2}\beta_{j,g}z^j,
 \qquad F_g(p,x)=U^{r/2-Hg}f_g(z).
\end{equation*}
The lower-case $f_g$ is the dimensionless version of the field polynomial $F_g$ in \eqref{eq:Ffield}; neither denotes a free energy. Using \eqref{eq:homderivative} and \eqref{eq:fieldvariation}, the primitive equation becomes
\begin{equation*}
 \frac{s/2-Hg}{\delta}\,q_g(z)
 -\frac{r e(z)}{2\delta e'(z)}q_g'(z)
 =s_g(z)+\frac{f_g(z)e''(z)-f_g'(z)e'(z)}{e'(z)^3}.
\end{equation*}
The common power $U^{-(r-1)/2-Hg}$ has cancelled. This is now an identity in $z$ with a known source and finitely many unknown numbers. Insert the polynomial expansions of $N_g$ and $f_g$, multiply by a common denominator, and equate the coefficient of every power of $z$. Each coefficient gives a linear equation for the $n_{g,\ell}$ and $\beta_{j,g}$. For instance, at $r=4$ and $g=1$ the bound gives $d_1^{\max}=10$: there are eleven primitive coefficients and three field coefficients, or fourteen unknown numbers before solving.

Solving this system determines the primitive and the fields simultaneously; any remaining coefficient equations test consistency. Alternatively, the descending-remainder method determines the same quantities by matching the strongest poles first. Agreement between these two arrangements is a useful implementation check. In both, the dispersive string equation is reserved for the subsequent comparison.

There is a useful distinction in the examples. Define $e_m(z)=2T_m(z/2)$, where $T_m(\cos\theta)=\cos(m\theta)$. A pure-Chebyshev pair satisfies
\begin{equation}\label{eq:Chebyidentity}
 s e_s e_r'-r e_r e_s'
 =\frac{2rs\sin((r-s)\theta)}{\sin\theta},
 \qquad z=2\cos\theta.
\end{equation}
Within the homogeneous single-flow ansatz \eqref{eq:hompair}, a nonzero constant bracket of this pure-Chebyshev form requires adjacent degrees. This explains the particularly simple unitary family used in Section~\ref{sec:unitary}.

The adjacent-degree condition belongs to this particular homogeneous single-flow ansatz. General coprime minimal strings also have Chebyshev descriptions at conformal times \cite{SeibergShih2004}. Our nonunitary examples use other leading polynomials and a specified choice of hierarchy times, while retaining a single moving scale $U$. They will therefore test the same reconstruction beyond the pure-Chebyshev family.

\subsection{A projector for several boundaries}\label{sec:boundaries}
The scalar resolvent was obtained by pairing a solution with its adjoint. To add boundaries, we retain the individual components of that pairing in matrix form. Introduce a column $v$ and a row $\ell$, indexed by $j=0,\ldots,r-1$:
\begin{equation}\label{eq:rowcolumn}
 v_j=\cP_j(w),\qquad
 \ell_j=\sum_{k=j+1}^r(-1)^{k-1-j}\Pi_{k-1-j}[c_k].
\end{equation}
The same finite sum as in \eqref{eq:Sigma} gives $\ell v=\Sigma$. Hence
\begin{equation}\label{eq:projector}
 \cM=\frac{v\ell}{\Sigma}=\Rhat v\ell,
 \qquad \cM^2=\cM,\quad \tr\cM=1,\quad \cM_{1r}=\Rhat.
\end{equation}
The normalization has turned the row and column into a rank-one projector in the defining $r$-dimensional representation. Its upper-right entry is the original scalar resolvent, so the matrix description inherits the normalization already fixed in the one-boundary calculation.

The derivative vector $\boldsymbol\psi=(\psi,D\psi,\ldots,D^{r-1}\psi)^t$ obeys a first-order companion system,
\begin{equation}\label{eq:companion}
 \h\partial_x\boldsymbol\psi=L\boldsymbol\psi,
 \qquad
 L=\begin{pmatrix}
 0&1&0&\cdots&0\\
 0&0&1&\cdots&0\\
 \vdots&&&\ddots&\vdots\\
 0&\cdots&0&0&1\\
 E-c_0&-c_1&\cdots&-c_{r-2}&-c_{r-1}
 \end{pmatrix}.
\end{equation}
The adjoint pairing gives the dual row system, so
\begin{equation}\label{eq:matrixODE}
 \h\partial_x\cM=[L,\cM].
\end{equation}
We can therefore pass between the scalar equation and a finite matrix system at every scalar order. Additional boundaries will be described by products of this same projector evaluated at different energies.

For example, at third order the projector is
\begin{equation}\label{eq:projector3}
 \cM=\Rhat
 \begin{pmatrix}1\\w\\w^2+\h w'\end{pmatrix}
 \begin{pmatrix}\widetilde w^2+\h\widetilde w'+b&-\widetilde w&1\end{pmatrix}.
\end{equation}
Expanding gives $\cM=\sum_{m\geq0}\h^mM_m$. The odd coefficients $M_1,M_3,\ldots$ generally survive, even though the scalar resolvent is even. They will contribute when projector coefficients are multiplied in multi-boundary amplitudes.

A specified string pair supplies a second differential equation for the wavefunction:
\begin{equation}\label{eq:spectralwave}
 \h\partial_E\psi=-\frac P\kappa\psi.
\end{equation}
The sign follows from \eqref{eq:stringpair}. Differentiating $(Q-E)\psi=0$ gives $(Q-E)\partial_E\psi=\psi$, while $(Q-E)P\psi=-\kappa\h\psi$. Thus \eqref{eq:spectralwave} is compatible with the original spectral equation.

To write it as a matrix system, reduce each $-D^jP/\kappa$ by right division by $Q-E$. The remainder has order less than $r$ and gives row $j+1$ of a matrix $\mathsf B$. This operation is often called folding the scalar operator. A formal fundamental matrix $\Psi$, whose columns are a basis of companion solutions, then satisfies
\begin{equation}\label{eq:flatness}
 \h\partial_x\Psi=L\Psi,\qquad
 \h\partial_E\Psi=\mathsf B\Psi,\qquad
 \partial_x\mathsf B-\partial_E L+\h^{-1}[\mathsf B,L]=0.
\end{equation}
The last equation is compatibility, or zero curvature. The construction is the folded reduced-KP system used in the differential-system approach \cite{BergereBorotEynard2015}.

Since $\partial_E L=e_{r1}$, where $e_{r1}$ has its only nonzero entry in row $r$, column $1$, equations~\eqref{eq:matrixODE} and \eqref{eq:flatness} give
\begin{equation}\label{eq:traceprimitive}
 \boxed{\displaystyle
 \partial_x\tr(\mathsf B\cM)=\Rhat,
 \qquad \Qhat_g=[\h^{2g}]\tr(\mathsf B\cM).}
\end{equation}
This gives a direct link between the two operators: the spectral connection from $P$ supplies a primitive for the resolvent built from $Q$. The identity holds as a formal series. Its short proof, and a scalar version expressed directly in WKB coefficients, are given in Appendix~\ref{app:proofs}.

To use this primitive in the inverse problem, we must still check that its positive-genus coefficients are rational in the required jets, have only the allowed poles and give a differential regular at infinity. Once these conditions hold, the uniqueness argument of Section~\ref{sec:reconstruction} identifies the reconstructed fields with the string solution at every genus. Appendix~\ref{app:proofs} states this implication precisely. It supplies existence for such string solutions, while leaving open the compatibility of an arbitrary leading family.

Choose a constant diagonal rank-one selector $\Pi$ for one formal sheet. Then $\cM=\Psi\Pi\Psi^{-1}$. The transport kernel
\begin{equation}\label{eq:transport}
 \cK(E_1,E_2)=\frac{\Psi^{-1}(E_1)\Psi(E_2)}{E_2-E_1}
\end{equation}
is built from the same pairing. For a dual row at $E_1$ and a column at $E_2$, \eqref{eq:Lagrange} gives $\partial_x\cK_{ij}=\chi_i(E_1)\psi_j(E_2)/\h$. Off the diagonal, the concomitant therefore integrates a product of solutions. This is the matrix form of the off-diagonal route to correlation kernels \cite{BergereBorotEynard2015}. Johnson develops a related off-diagonal differential equation for the kernel integrand in the Schr\"odinger case, with the diagonal limit reducing to Gel'fand--Dikii \cite{Johnson2025Kernel}.

Its finite diagonal term defines the first cumulant,
\begin{equation}\label{eq:C1}
 \cC_1(E)=\tr\!\left(\Pi\Psi^{-1}\partial_E\Psi\right)
        =\frac1\h\tr(\mathsf B\cM),
 \qquad \partial_x\cC_1=\frac{\Rhat}{\h}.
\end{equation}
Taking the coefficient of $\h^{2g-1}$ therefore returns the same one-boundary primitive as before. This is the normalization link that allows the higher cumulants to be interpreted as additional-boundary amplitudes.

For $n\geq2$, the connected determinant expansion retains single cycles and gives
\begin{equation}\label{eq:cycles}
 \boxed{\displaystyle
 \cC_n(E_1,\ldots,E_n)=
 -\sum_{\sigma\in\operatorname{Cyc}_n}
 \frac{\tr\!\left(\cM_1\cM_{\sigma(1)}\cdots
                         \cM_{\sigma^{n-1}(1)}\right)}
 {\prod_{i=1}^n(E_i-E_{\sigma(i)})}.}
\end{equation}
Here $\cM_i=\cM(E_i,x;\h)$, and $\operatorname{Cyc}_n$ consists of the $(n-1)!$ permutations with one $n$-cycle. Each term visits all $n$ energies once, and each cycle is counted once. We use the ordinary $r\times r$ matrix trace; some normalized Drinfeld--Sokolov generating functions use an adjoint-representation trace \cite{BertolaDubrovinYang2016}.

Our kernel denominator in \eqref{eq:transport} is the reverse of that in equation~(2-1) of BBE \cite{BergereBorotEynard2015}. The constant minus sign in \eqref{eq:cycles} is the corresponding choice after writing its denominators as $E_i-E_{\sigma(i)}$. For example,
\begin{align}
 \cC_2&=\frac{\tr(\cM_1\cM_2)}{(E_1-E_2)^2},\label{eq:C2}\\
 \cC_3&=-\frac{\tr(\cM_1[\cM_2,\cM_3])}
 {(E_1-E_2)(E_2-E_3)(E_3-E_1)}.\label{eq:C3}
\end{align}
Keeping the sign and the denominator orientation together is necessary for comparison with the scalar primitive.

For a topological-type solution, the expansion is
\begin{equation}\label{eq:TT}
 \cC_n=\sum_{g\geq0}\h^{2g-2+n}C_{g,n}.
\end{equation}
The genus coefficients are then pulled back from energy to curve coordinates:
\begin{equation}\label{eq:Wcycles}
 W_{g,n}(p_1,\ldots,p_n)=
 \left.\left(\prod_{i=1}^n\Delta(p_i,x)\right)
 [\h^{2g-2+n}]\cC_n(E_1,\ldots,E_n)
 \right|_{x=\mu,\ E_i=E(p_i,\mu)}.
\end{equation}
The factors of $\Delta_i$ simply convert $\dd E_i$ to $\dd p_i$. For this interpretation, the differential system must have topological type: lower forbidden powers of $\h$ cancel, the remaining powers have the correct parity, and the coefficients have the required poles. These properties hold for the applicable reduced-KP systems studied in BBE \cite{BergereBorotEynard2015}. The general loop equations and the sufficient conditions for topological type explain how this identification extends to other differential systems \cite{Belliard:2016mxj,Belliard:2016cyc}. Applying that extension requires checking the hypotheses for the system under study.

A formal insertion can be defined on the fundamental matrix by
\begin{equation}\label{eq:insertion}
 \delta_{E_0}\Psi(E)=
 \left(\frac{\cM(E_0)}{E_0-E}+\mathsf G(E_0)\right)\Psi(E).
\end{equation}
The compensating gauge matrix $\mathsf G(E_0)$ is independent of the argument $E$. It cancels from cyclic traces. Differentiating the kernel gives
\begin{equation}
 \delta_{E_0}\cK(E_1,E_2)
 =-\cK(E_1,E_0)\Pi\cK(E_0,E_2),
 \qquad \delta_{E_0}\cC_n=\cC_{n+1}.
\end{equation}
This explains the name boundary insertion: differentiating a connected cumulant produces the next one. The operation acts on the full hierarchy of fields and explores its independent deformations. Its relation to the hierarchy times extends the KdV loop-operator construction \cite{BanksDouglasSeibergShenker1990,OkuyamaSakai2019,Lowenstein2024OpenClosed,Johnson2026Universal,BertolaDubrovinYang2016}. The Boussinesq flow \eqref{eq:Boussinesq} gives a concrete reason to retain that full dependence.

In Johnson's universal KdV formulation the loop operator acts on the leading potential and its derivatives before the background is specialized \cite{Johnson2026Universal}. The projector formula likewise retains the full field dependence during insertion; we specialize to the homogeneous background afterwards.

The genus-one two-boundary function gives a useful example of the required coefficient bookkeeping. From \eqref{eq:C2},
\begin{equation}\label{eq:W12matrix}
 W_{1,2}(p_1,p_2)=
 \frac{\Delta_1\Delta_2}{(E_1-E_2)^2}
 \sum_{m=0}^2\tr\!\left(M_m(E_1)M_{2-m}(E_2)\right).
\end{equation}
In particular, the product $M_1(E_1)M_1(E_2)$ gives an essential contribution from the odd coefficients. The same scalar data through $\h^2$ also give $W_{0,4}$, while $W_{0,3}$ needs only $M_0$ and $M_1$. Once the projector is available, the number of boundaries tells us which products and which coefficient to extract.

For the unstable cylinder we use the Bergman differential in \eqref{eq:spectraldata}. A connected matrix resolvent is sometimes defined by subtracting an energy-space Cauchy term, giving
\begin{equation}\label{eq:cylinder}
 \omega^{\rm conn}_{0,2}=B(p_1,p_2)
       -\frac{\dd E_1\dd E_2}{(E_1-E_2)^2}.
\end{equation}
That choice must be tracked in a comparison of unstable amplitudes. The stable functions displayed in this paper use the unsubtracted Bergman initial data.

Finally, the $r$ mutually orthogonal sheet projectors sum to the identity. Their relation to the $r-1$ nontrivial primary time sectors is expressed through expansions at spectral infinity and the appropriate normalization of the hierarchy times. The explicit higher-Airy dictionary in Section~\ref{sec:airy} is one instance where that relation can be written coefficient by coefficient.

For one boundary, compute the scalar series through $\h^{2g}$ and use the remainder construction, or extract \eqref{eq:traceprimitive} from a specified string pair. For $n\geq2$, set $m=2g-2+n$. It is sufficient to know the even fields through $\h^{2\lfloor m/2\rfloor}$, retain all ordering contributions through $\h^m$, and construct $M_0,\ldots,M_m$. Formula~\eqref{eq:Wcycles} then gives the desired answer directly.

The order of operations matters: take all derivatives first, then set $x=\mu$. The calculation requires both the fields and their derivatives at the endpoint. The homogeneous derivative \eqref{eq:homderivative} keeps this background dependence explicit throughout the calculation. The procedure is finite at fixed $g,n$, with a number of cycles that grows with $n$. Comparing its computational cost with topological recursion would require a separate benchmark.

\subsection{Comparison with spectral recursion}\label{sec:EO}
We now have two ways to reach the amplitudes. The scalar construction starts from the differential operator. Eynard--Orantin recursion starts directly from the spectral data \eqref{eq:spectraldata} \cite{EynardOrantin2007}. Comparing their answers is a useful test because the intermediate calculations are different.

For the ordinary simple-ramification formulas below, the disk must be regular at each finite ramification point and satisfy $y_p(\alpha)\ne0$. In the examples here this follows directly from the moving disk primitive. Multiplying $\cD\Qhat_0=1/\Delta$ by $\Delta$ gives
\begin{equation*}
 \Delta\left.\partial_x\Qhat_0\right|_p
       -V\partial_p\Qhat_0=1.
\end{equation*}
Assume that $\Qhat_0$ and its fixed-$p$ background derivative are regular near the point in question. Taking the limit to $p=\alpha$ and then setting $x=\mu$ gives
\begin{equation*}
 y_p(\alpha)=-\frac{1}{v_\alpha},\qquad
 v_\alpha:=V(\alpha,\mu).
\end{equation*}
Thus the nonzero branch-value velocities of the reconstruction domain also ensure the nonvanishing disk derivatives needed by the residue formulas. This argument uses regularity of the disk primitive in addition to the conditions on $\Delta$ and $V$; those conditions alone do not prove the required regularity for an arbitrary leading family. Adding a function $C(E)$ that is regular at the branch value does not change this relation, since its $p$ derivative is $C'(E)\Delta$.

Near each simple ramification point $\alpha$, let $\sigma_\alpha(p)$ be the other local point with the same energy:
\begin{equation}
 E(\sigma_\alpha(p))=E(p),\qquad \sigma_\alpha(\alpha)=\alpha.
\end{equation}
The two conjugate points meet at $\alpha$. The kernel is
\begin{equation}\label{eq:EOkernel}
 K_\alpha(p_0,p)=
 \frac{(p_0-p)^{-1}-(p_0-\sigma_\alpha(p))^{-1}}
 {2[y(p)-y(\sigma_\alpha(p))]E_p(p)}\frac{\dd p_0}{\dd p}.
\end{equation}
The recursion combines amplitudes of lower complexity on these two local sheets. Its first term lowers the genus, while its product term distributes the genus and the remaining boundaries between two pieces. With $J=(p_2,\ldots,p_n)$, it reads
\begin{align}\label{eq:EOrec}
 \omega_{g,n}(p_1,J)=\sum_\alpha\Res_{p=\alpha}K_\alpha(p_1,p)
 \bigg[&\omega_{g-1,n+1}(p,\sigma_\alpha(p),J)\nonumber\\
 &+\sum_{\substack{g_1+g_2=g\\ I\sqcup I^c=J}}'
 \omega_{g_1,|I|+1}(p,I)
 \omega_{g_2,|I^c|+1}(\sigma_\alpha(p),I^c)\bigg].
\end{align}
The prime excludes factors $\omega_{0,1}$; the Bergman differential is retained as initial data. The overall sign is fixed by the kernel as written in \eqref{eq:EOkernel}. For example, $E=p^2/2$ and $y=p$ give
\begin{equation}\label{eq:EOsign}
 W_{1,1}=-\frac1{8p^4},\qquad W_{2,1}=-\frac{105}{128p^{10}}.
\end{equation}
Reversing $y$ multiplies stable outputs by $(-1)^{2g-2+n}$. More generally, multiplying $y$ by a nonzero constant $a$ multiplies them by $a^{-(2g-2+n)}$. This is why matching the disk precedes all positive-genus comparisons.

The first residues can be written in terms of the local derivatives of $E$ and $y$. For three planar boundaries,
\begin{equation}\label{eq:EO03}
 W_{0,3}(p_1,p_2,p_3)
 =-\sum_{\alpha:E_p(\alpha)=0}
   \frac{1}{E_{pp}(\alpha)y_p(\alpha)}
   \prod_{i=1}^3\frac1{(p_i-\alpha)^2}.
\end{equation}
The full genus-one answer is a finite sum of the local jets of $E$ and $y$. Appendix~\ref{app:proofs} gives the complete formula \eqref{eq:EO11} and its residue derivation.

The same recursion constructs $\omega_{1,2}$ from $\omega_{0,3}$, $\omega_{1,1}$ and $B$. Its genus-two one-point source is then
\begin{equation}\label{eq:EO21source}
 \omega_{1,2}(p,\sigma_\alpha(p))
 +\omega_{1,1}(p)\omega_{1,1}(\sigma_\alpha(p)).
\end{equation}
Thus the spectral calculation of $\omega_{2,1}$ uses $\omega_{1,2}$ along the way. The scalar primitive gives the same one-boundary quantity directly. Equality of the full rational functions compares two different constructions, and tests more than the expected orders of their poles.

For later examples it is useful to evaluate the residues directly from the ramification polynomial. Some nonunitary curves have complex roots, and one dual projection has an irreducible quartic ramification polynomial. Both can be treated exactly. If $E=e(z)$ at the endpoint, set $\mathscr D=e'/r$ and work in
\begin{equation}\label{eq:quotient}
 \mathscr A=\mathbb Q[a]/(\mathscr D(a)).
\end{equation}
Within this quotient algebra, $a$ denotes a formal root variable. Squarefreeness makes $\mathscr A$ a finite reduced algebra, possibly a product of number fields. The algebra trace satisfies
\begin{equation}
 \operatorname{Tr}_{\mathscr A}f(a)
 =\sum_{\mathscr D(\alpha)=0}f(\alpha).
\end{equation}
Thus local residues can be reduced modulo $\mathscr D$ and summed exactly within the quotient algebra. This is also useful for keeping all external coordinates symbolic in a two-boundary comparison.

The residue calculation uses only $E$, $y$, $B$ and the local involutions, so it is independent of the reconstructed fields and primitive numerators. The string-equation comparison asks a different question: starting with those fields, does an independently constructed fractional-power partner have the required commutator? Together these calculations test both the background and its amplitudes.

\section{Pure gravity and Ising}\label{sec:ising}
\subsection{Pure gravity as a calibration}
We begin with pure gravity, where the expected string equation is simple enough to recognize at the end of the calculation. Using its third-order realization also gives a first test with two fields. Choose $(r,s)=(3,2)$, so
\begin{equation}\label{eq:PIcurve}
 U=x^{1/2},\qquad e(z)=z^3-3z,\qquad
 \pi(z)=z^2-2,\qquad \kappa=-3.
\end{equation}
This gives the endpoint disk $y=(z^2-2)/3$. We take the leading fields from this family and leave both quantum corrections free. At genus one, \eqref{eq:genericB}--\eqref{eq:genericA} give $A=0$ and $B=1/(16x^2)$. Continuing the same primitive calculation through genus three yields
\begin{align}\label{eq:PIpotential}
 Q&=D^3-\frac32\{\mathfrak u,D\},\nonumber\\
 \mathfrak u&=x^{1/2}-\frac{\h^2}{48x^2}
 -\frac{49\h^4}{4608x^{9/2}}
 -\frac{1225\h^6}{55296x^7}+O(\h^8).
\end{align}
The zeroth-order coefficient field $a$ remains zero at each calculated genus, as a result of the reconstruction. To identify the solution, substitute the series for $\mathfrak u$ into
\begin{equation}\label{eq:PI}
 \mathfrak u^2-\frac{\h^2}{6}\mathfrak u''=x
\end{equation}
The residual vanishes through $\h^6$. We have recovered Painlev\'e~I in this normalization, with its genus expansion obtained from the primitive conditions \cite{DouglasShenker1990,GrossMigdal1990}.

BBE's third-order pure-gravity example uses $t_B=3x$, $\h_B=3\h$ and the equation $3\mathfrak u^2-\h_B^2\mathfrak u_{t_Bt_B}/2=t_B$ \cite[Section~6.2]{BergereBorotEynard2015}. The equality $\h_B\partial_{t_B}=D$ identifies the differential operators, while the disk normalization fixes the scale of the amplitudes. The endpoint one-boundary result is
\begin{equation}\label{eq:PIW11}
 W^{(3,2)}_{1,1}(p)
 =-\frac{p^6-p^4+19p^2+5}{48(p^2-1)^4}.
\end{equation}
Comparison with the literature gives agreement for this function and $W^{(3,2)}_{2,1}$ after the respective normalization factors $3$ and $3^3$. Independent EO calculations reproduce both functions as well. The full genus-one two-boundary function also agrees between the matrix and EO constructions; one compact specialization is
\begin{equation}\label{eq:PIW12}
 W^{(3,2)}_{1,2}(0,p)=
 \frac{10p^{10}-30p^8+45p^6+3p^4+69p^2+23}
 {48(p^2-1)^6}.
\end{equation}
At genus three, the field and primitive are checked through the string equation and trace identity. The independent EO comparisons extend through genus two.

\subsection{The Ising background and its one-boundary amplitudes}

The next example asks what changes when matter is added while the endpoint energy polynomial is kept the same. The Ising model on a fluctuating surface has a classical matrix-model realization \cite{Kazakov1986}; for comparison with the literature we use the differential-system conventions of Berg\`ere, Borot and Eynard \cite{BergereBorotEynard2015}.
Now take $(r,s)=(3,4)$ on the homogeneous zero-magnetic-field sector,
\begin{equation}\label{eq:Isingleading}
 a_0=0,\qquad b_0=-3U,\qquad U=x^{1/3}.
\end{equation}
At $x=1$ we again have $E=p^3-3p$, just as in pure gravity. Their different powers of $x$ give different motions away from the endpoint. We can see the consequence by deriving the Ising disk directly.

Using the fixed-energy integral \eqref{eq:diskintegral}, we eliminate the background through $b_0=E/p-p^2$ and $x=-b_0^3/27$. Thus $V=b_0'p=-9p/b_0^2$, and \eqref{eq:diskintegral} becomes an elementary rational integral. Its primitive is
\begin{equation}\label{eq:IsingQ0}
 \Qhat_0(E,p)=-\frac{E^2}{18p^2}-\frac{2Ep}{9}
                 +\frac{p^4}{36}+C(E).
\end{equation}
Only after this integration do we set $x=1$ and $E=p^3-3p$. Dropping the common $C(E)$ gives
\begin{equation}\label{eq:Isingdisk}
 E(p)=p^3-3p,\qquad
 y(p)=-\frac14(p^4-4p^2+2).
\end{equation}
The same result follows from \eqref{eq:homdisk} with $\pi=e_4$ and $\kappa=4$. The quartic disk records the background motion, supplying information that complements the common endpoint energy polynomial. Its normalization has been fixed by integrating the leading resolvent.

At genus one, $b_0(b_0')^2=-3/x$, so \eqref{eq:genericB} immediately gives the correction $\h^2/(8x^2)$ to $b$. The correction to $a$ vanishes. For higher genus we keep the entire field $a$ zero on the zero-magnetic-field sector and write
\begin{equation}\label{eq:Isingansatz}
 b=-3U+\sum_{g\geq1}d_g\h^{2g}U^{1-7g},
 \qquad
 \Qhat_g=U^{2-7g}q_g(z),\quad z=p/\sqrt U.
\end{equation}
The derivative needed in the scalar recurrence is
\begin{equation}\label{eq:Isingderivative}
 \cD\bigl(U^\nu f(z)\bigr)
 =U^{\nu-3}\left[
 \frac\nu3f(z)-\frac{z(z^2-3)}{6(z^2-1)}f'(z)\right].
\end{equation}
This change of variables makes each fixed-energy derivative explicit. At each genus we then determine the unknown coefficients by solving a rational identity in $z$.

The allowed poles are at $z=\pm1$. Evenness in $z$ and regularity of $3(z^2-1)q_g(z)\dd z$ at infinity give
\begin{equation}\label{eq:Isingqansatz}
 q_g(z)=\frac{\sum_{j=0}^{6g-3}c_{g,j}z^{2j}}
                     {(z^2-1)^{6g-1}}.
\end{equation}
At genus $g$ the unknowns are the field coefficient $d_g$ and the coefficients of the primitive numerator. There are $5$, $11$ and $17$ unknowns at the first three genera, respectively. In each case the linear equations are consistent and determine them uniquely, giving
\begin{equation}\label{eq:Isingb}
 b=-3x^{1/3}+\frac{\h^2}{8x^2}
 +\frac{1925\h^4}{15552x^{13/3}}
 +\frac{509575\h^6}{1119744x^{20/3}}+O(\h^8).
\end{equation}
These coefficients have been obtained solely from the primitive conditions within the stated homogeneous Ising sector. We can now ask whether they satisfy the string equation of that model.

For the independent comparison, put $b=-3\mathfrak u$. The complete BBE Ising pair is recorded in Appendix~\ref{app:published}. With $t_B=4x$, $\h_B=4\h$, one has $D_B=D$. Their scalar equation becomes
\begin{equation}\label{eq:Isingstring}
 \mathfrak u^3-\frac34\h^2\mathfrak u\mathfrak u''
 -\frac38\h^2(\mathfrak u')^2+\frac{\h^4}{24}\mathfrak u''''=x.
\end{equation}
Substitution of \eqref{eq:Isingb} makes the residual vanish through $\h^6$. The normalization map also identifies the complete ordered operator, including the derivative contribution, so the comparison with the literature tests the quantum equation as well as its classical limit \cite[Section~6.3]{BergereBorotEynard2015}.

The reconstructed primitive gives
\begin{equation}\label{eq:IsingW11}
 W_{1,1}(p)=-\frac{3p^6-3p^4+17p^2+7}{72(p^2-1)^4}.
\end{equation}
At genus two,
\begin{equation}\label{eq:IsingW21}
 W_{2,1}(p)=-\frac{5P_2(p)}{31104(p^2-1)^{10}},
\end{equation}
where the numerator $P_2$ is written in Appendix~\ref{app:polynomials}.

The genus-three answer is
\begin{equation}\label{eq:IsingW31}
 W_{3,1}(p)=-\frac{5P_3(p)}{10077696(p^2-1)^{16}},
\end{equation}
with $P_3$ in Appendix~\ref{app:polynomials}. The increasing denominator powers are consistent with the primitive bound, while the degree of each numerator makes the differential regular at infinity.

The disk in \eqref{eq:Isingdisk} is one quarter of BBE's disk, so their stable differentials and ours are related by
\begin{equation}\label{eq:BBEscale}
 \omega_{g,n}=4^{2g-2+n}\omega^B_{g,n}.
\end{equation}
After this rescaling, comparison with the literature reproduces \eqref{eq:IsingW11} and every coefficient of the genus-two and genus-three numerators, with factors $4^3$ and $4^5$ at those genera \cite[Section~6.3]{BergereBorotEynard2015}. Recovering these known amplitudes tests both the primitive criterion and the normalization inherited from the disk.

\subsection{Additional Ising boundaries}
We now use the same scalar data to add boundaries. The two branch points suggest a convenient way to write the answer: let $L_\alpha(p)=(p-\alpha)^{-1}$ for $\alpha=\pm1$. The three-cycle formula gives
\begin{equation}\label{eq:IsingW03}
 W_{0,3}(p_1,p_2,p_3)=-\frac16
       \sum_{\alpha=\pm1}\prod_{i=1}^3L_\alpha(p_i)^2.
\end{equation}
The denominators at coincident energies cancel when the cycles are combined, leaving the final answer regular there. The result agrees with the direct residue formula \eqref{eq:EO03} and, after the disk rescaling, with the three-point function in the literature.

For the genus-one two-boundary function, use the symmetric polynomial $\mathsf H_\alpha$ defined in Appendix~\ref{app:isingmulti}. The full result is
\begin{align}\label{eq:IsingW12}
 W_{1,2}(p_1,p_2)=\frac1{3456}\sum_{\alpha=\pm1}\big[&
 L_\alpha(p_1)^2L_\alpha(p_2)^2
 \mathsf H_\alpha(L_\alpha(p_1),L_\alpha(p_2))\nonumber\\
 &+3L_\alpha(p_1)^2L_{-\alpha}(p_2)^2\big].
\end{align}
The last term couples the two ramification points and records their contribution to a single global amplitude. The complete function is symmetric in its arguments and regular when the two points coincide away from ramification. For example,
\begin{equation}\label{eq:IsingW12slice}
 W_{1,2}(0,p)=
 \frac{35p^{10}-105p^8+160p^6-36p^4+117p^2+69}
 {216(p^2-1)^6},
 \qquad W_{1,2}(0,0)=\frac{23}{72}.
\end{equation}
An independent EO calculation reproduces \eqref{eq:IsingW12} in both variables. This adds a complete two-boundary comparison to the one-point comparisons with the literature above.

Appendix~\ref{app:isingmulti} gives two further outputs with more limited checks. The planar four-boundary expression, obtained from exact local residues, agrees with the matrix cycles at eight exact rational tuples. The specialization $W_{2,2}(0,p)$ is reproduced by a separate projector extraction, and a local expansion checks the six most singular Laurent coefficients, those of $(p-\alpha)^{-12},\ldots,(p-\alpha)^{-7}$ at each $\alpha=\pm1$. A full four-variable comparison for $W_{0,4}$ and a full two-variable matrix/EO comparison for $W_{2,2}$ remain to be done.

\section{From unitary to nonunitary models}\label{sec:unitary}
\subsection{The moving unitary family}
The Ising curve belongs to a moving family whose two operator orders are adjacent. This gives a natural way to increase the order while keeping the leading background and disk explicit. For $(r,s)=(r,r+1)$, the Chebyshev construction is
\begin{equation}\label{eq:Chebycurve}
 U=x^{1/r},\qquad E(p,x)=U^{r/2}e_r(p/\sqrt U),\qquad
 y(p)=-\frac{e_{r+1}(p)}{r+1}\quad(x=1).
\end{equation}
The identity \eqref{eq:Chebyidentity} gives $\kappa=r+1$. The case $r=3$ is precisely the Ising disk just derived. The cases $(4,5)$ and $(5,6)$ have central charges $7/10$ and $4/5$, respectively; the former is tricritical Ising, while the latter here denotes the diagonal $A$-series minimal model. Throughout these examples we use the diagonal $A$-series modular invariant \cite{DiFrancescoGinspargZinnJustin1995}.

The finite ramification points are
\begin{equation}\label{eq:Chebyroots}
 \alpha_j=2\cos\frac{j\pi}{r},\qquad j=1,\ldots,r-1.
\end{equation}
At the endpoint these are the ramification coordinates. Along the moving family, the branch values are $2(-1)^j\sqrt x$ and their velocities are $(-1)^j/\sqrt x$. For $x\ne0$ the ramification remains simple and every velocity is nonzero. Both conditions continue to hold when several points project to the same energy, so the reconstruction applies throughout this regular domain.

Substituting the Chebyshev derivatives into \eqref{eq:EO03} gives the formula valid at arbitrary scalar order
\begin{equation}\label{eq:Cheby03}
 W^{\mathrm C,r}_{0,3}(p_1,p_2,p_3)
 =-\frac2{r^2}\sum_{j=1}^{r-1}\sin^2\frac{j\pi}{r}
              \prod_{i=1}^3\frac1{(p_i-\alpha_j)^2}.
\end{equation}
The superscript $\mathrm C$ specifies this Chebyshev background. At $r=3$, each branch coefficient is $-1/6$, recovering \eqref{eq:IsingW03}.

There is also a closed genus-one formula. The differential equation
\begin{equation}
 (4-p^2)e_r''-pe_r'+r^2e_r=0
\end{equation}
reduces the jets entering \eqref{eq:EO11} to rational expressions in $\alpha_j$. For example, $X_3/X_2=3\alpha_j/(4-\alpha_j^2)$ and $Y_2/Y_1=-r\alpha_j/(4-\alpha_j^2)$. The complete result is
\begin{align}\label{eq:Cheby11}
 W^{\mathrm C,r}_{1,1}(p)=\sum_{j=1}^{r-1}\bigg[&
 -\frac{4-\alpha_j^2}{16r^2(p-\alpha_j)^4}
 +\frac{\alpha_j}{16r^2(p-\alpha_j)^3}\nonumber\\
 &+\left(\frac{\alpha_j^2}{16r^2(4-\alpha_j^2)}
       -\frac{(r+2)(r-1)}{48r^2}\right)\frac1{(p-\alpha_j)^2}
 \bigg].
\end{align}
The branch sums give the full answer at arbitrary scalar order. In particular, they retain the lower-order poles as well as the leading local Airy singularity. This makes them useful for comparing the complete amplitudes obtained by reconstructing the fields.

The large-$p$ term of \eqref{eq:Cheby11} is $-(r-1)/(48p^2)$. Restoring the $x$ dependence and comparing the leading term of the primitive derivative with \eqref{eq:fieldvariation} gives
\begin{equation}\label{eq:Chebytop}
 u_{r-2,1}=\frac{r(r-1)}{48x^2},\qquad
 F_1^{\rm loc}=-\frac{r-1}{48}\log x,
 \qquad u_{r-2,1}=r\partial_x^2F_1^{\rm loc}.
\end{equation}
The simple expression determines the highest nonleading field on this unitary family. To see what happens to the other fields, we now carry out the full reconstruction at orders four and five.

At fourth order, the full homogeneous reconstruction gives
\begin{align}\label{eq:fields45}
 u_2&=-4U+\frac{\h^2}{4}U^{-8}
              +\frac{3549\h^4}{8192}U^{-17}+O(\h^6),\nonumber\\
 u_1&=0+O(\h^6),\nonumber\\
 u_0&=2U^2-\frac{9\h^2}{32}U^{-7}
              -\frac{18477\h^4}{40960}U^{-16}+O(\h^6),
 \qquad U=x^{1/4}.
\end{align}
All three field corrections were left free, so the zero in $u_1$ is part of the answer. The genus-one amplitude is
\begin{equation}\label{eq:W11r4}
 W^{\mathrm C,4}_{1,1}(p)=
 -\frac{2p^{10}-9p^8+17p^6+2p^4-4p^2+8}
 {32p^4(p^2-2)^4}.
\end{equation}
Its poles occur at the three roots of $e_4'$, and it agrees with \eqref{eq:Cheby11}. The genus-two result is
\begin{equation}\label{eq:W21r4}
 W^{\mathrm C,4}_{2,1}(p)
 =-\frac{N_4(p)}{40960p^{10}(p^2-2)^{10}},
\end{equation}
with $N_4$ given in Appendix~\ref{app:polynomials}.

At fifth order the corresponding fields are
\begin{align}\label{eq:fields56}
 u_3&=-5U+\frac{5\h^2}{12}U^{-10}
              +\frac{1309\h^4}{1125}U^{-21}+O(\h^6),\nonumber\\
 u_2&=u_0=0+O(\h^6),\nonumber\\
 u_1&=5U^2-\h^2U^{-9}
              -\frac{46003\h^4}{18000}U^{-20}+O(\h^6),
 \qquad U=x^{1/5}.
\end{align}
The vanishing corrections to $u_2,u_0$ are outputs of the all-field fit. The endpoint amplitudes include
\begin{align}\label{eq:W11r5}
 W^{\mathrm C,5}_{1,1}(p)
 &=-\frac{10p^{14}-81p^{12}+257p^{10}-372p^8+325p^6-228p^4+135p^2+33}
 {120(p^4-3p^2+1)^4},\\
 W^{\mathrm C,5}_{2,1}(p)
 &=-\frac{N_5(p)}{180000(p^4-3p^2+1)^{10}}.\label{eq:W21r5}
\end{align}
The numerator $N_5$ is recorded in Appendix~\ref{app:polynomials}. The lower fields in \eqref{eq:fields45} and \eqref{eq:fields56} have their own nontrivial corrections. Their values show why recovering the highest field alone would leave part of the operator undetermined.

At both scalar orders we compare the reconstructed operator with the standard string-pair construction by independently forming $P=(Q^{(r+1)/r})_+$. Its commutator satisfies $[P,Q]-(r+1)\h=0$ through $\h^5$. We also calculate the amplitudes directly from the endpoint curve using EO recursion: the full $W_{1,1}$, $W_{2,1}$ and two-variable $W_{1,2}$ agree. The relation to the literature is through the hierarchy and its string-pair construction; the explicit amplitude comparison is with this independent residue calculation.

The same calculation has been performed at $r=6,7$ through genus two, with all $r-1$ homogeneous fields free. The simultaneous systems have full ranks $24,54$ at order six and $29,65$ at order seven for genera one and two. The descending-remainder calculation recovers the same coefficients independently of the simultaneous fit. Appendix~\ref{app:higherfields} gives the full nonzero fields.

At orders six and seven, the fractional-power string equations again hold through $\h^5$, and the folded trace primitives agree through genus two. The genus-one amplitudes reproduce \eqref{eq:Cheby11}, and the full genus-two one-boundary functions agree with EO. These comparisons extend the one-boundary check to additional fields and branch points. The full genus-one two-boundary comparison remains open for the forward $(6,7)$ and $(7,8)$ backgrounds.

\subsection{Dual scalar projections}

A further test comes from exchanging the two operators. This is the usual $p$--$q$ duality of the minimal-string description \cite{FukumaKawaiNakayama1992,SeibergShih2004,ChanIrieNiednerYeh2016}. It lets us describe one string solution through two different scalar equations, while asking how the choice of projection affects a marked boundary.
A shared string pair has a simple relation between its highest coefficient fields. From \eqref{eq:root},
\begin{equation}
 Q^{1/r}=D+\frac{u_{r-2}}rD^{-1}+O(D^{-2}),
 \qquad [D^{s-2}](Q^{s/r})_+=\frac sr u_{r-2}.
\end{equation}
Derivatives of that root coefficient contribute at lower powers of $D$. Therefore the two scalar realizations of the same normalized pair have
\begin{equation}\label{eq:dualtop}
 \mathfrak u=-\frac{u^{(r,s)}_{r-2}}r
             =-\frac{u^{(s,r)}_{s-2}}s.
\end{equation}
The equality holds for the same normalized pair, with the same $x$, $\h$ and background. It gives a useful common quantity to compare when the orders are exchanged.

If $Q^\vee=P$ and $P^\vee=Q$, then $\kappa^\vee=-\kappa$. The endpoint spectral data transform as
\begin{equation}\label{eq:swap}
 (E^\vee,y^\vee)=\left(-\kappa y,\frac E\kappa\right).
\end{equation}
The exchange changes the energy projection, its ramification points and the disk. It therefore changes the marked-boundary amplitudes. We first compare the complete dual operators and then calculate each projection's boundary functions using its own spectral data.

For fourth-order Ising, the leading data are
\begin{equation}
 E=p^4-4p^2+2,\qquad y=\frac14(p^3-3p).
\end{equation}
The reconstructed operator is the fourth-order member of \eqref{eq:BBEIsingops}, written in our variables:
\begin{equation}\label{eq:IsingdualQ}
 Q^{(4,3)}=D^4-4\mathfrak uD^2-4\h\mathfrak u'D
          +2\mathfrak u^2-\frac53\h^2\mathfrak u''.
\end{equation}
In symmetric fields this means $u_2=-4\mathfrak u$, $u_1=0$ and $u_0=2\mathfrak u^2+\h^2\mathfrak u''/3$. Comparison with the Ising pair in the literature includes the last derivative term and holds through the calculated order. The new boundary projection gives
\begin{equation}\label{eq:IsingdualW11}
 W^{(4,3)}_{1,1}(p)=
 -\frac{3p^{10}-15p^8+35p^6+22p^4-28p^2+24}
 {72p^4(p^2-2)^4}.
\end{equation}
The poles are governed by the same degree-four energy polynomial as in \eqref{eq:W11r4}. The disk and background scaling enter through the numerators, giving distinct amplitudes for the two backgrounds.

The $(5,4)$ and $(6,5)$ reconstructions likewise agree with the complete dual operators of the forward $(4,5)$ and $(5,6)$ solutions through genus two. In particular, the sixth-order highest field has
\begin{equation}\label{eq:dual65}
 \beta^{(6,5)}_{4,1}=\frac12,\qquad
 \beta^{(6,5)}_{4,2}=\frac{2618}{1875}.
\end{equation}
The independent reconstruction gives these values, in agreement with the shared-potential relation. For all three dual backgrounds, the full $W_{1,1}$, $W_{2,1}$ and two-variable $W_{1,2}$ also agree with EO. Appendix~\ref{app:NUdata} records the complete genus-two one-boundary expressions.

\subsection{Nonunitary pairs}\label{sec:nonunitary}
We next move away from the unitary family. The $(3,5)$ model has $c=-3/5$ and gives a useful example in which both leading fields move and the ramification points are complex. Choose
\begin{align}\label{eq:35curve}
 e(z)&=z^3+9z+9,\nonumber\\
 \pi(z)&=z^5+15z^3+15z^2+45z+90,\nonumber\\
 \kappa&=-405,\qquad U=x^{2/7},\qquad y(z)=\frac{\pi(z)}{405}.
\end{align}
The relation $3e\pi'-5\pi e'=-2835$ verifies the constant bracket before any quantum calculation. The roots $z=\pm\ii\sqrt3$ of $e'$ are simple, and their branch values have nonzero velocities. The polynomial-remainder algorithm therefore applies under the same analytic conditions.

Writing the operator as $Q=D^3+\{b,D\}/2+a$, the all-field reconstruction gives
\begin{align}\label{eq:35fields}
 b&=9U+\frac{\h^2}{7}U^{-7}
             -\frac{3289\h^4}{108045}U^{-15}+O(\h^6),\nonumber\\
 a&=9U^{3/2}+\frac{1573\h^4}{36015}U^{-29/2}+O(\h^6).
\end{align}
Here the genus-one correction to $a$ vanishes and its genus-two correction is nonzero. The example illustrates why each field must be determined afresh at each genus.

The independent model identification comes from the third-/fifth-order Baker--Akhiezer system of Chan, Irie, Niedner and Yeh \cite[Appendix~B.7]{ChanIrieNiednerYeh2016}. After the ordering and coordinate changes recorded in Appendix~\ref{app:published}, its two integrated string equations are
\begin{align}\label{eq:35strings}
 a^2&=\frac{b^3}{9}+\frac{\h^2}{3}bb''
      +\frac{\h^2}{4}(b')^2+\frac{\h^4}{15}b'''',\nonumber\\
 ab^2&+\frac{\h^2}{2}\left(b'a'+ab''+2ba''\right)
       +\frac{\h^4}{5}a''''=729x.
\end{align}
At leading order the equations give $a^2=b^3/9$ and $ab^2=729x$, consistent with \eqref{eq:35curve}. The reconstructed fields satisfy both through $\h^4$. Thus the comparison with the literature checks the two field equations together, including the correction to the lower field.

The genus-one one-boundary function is
\begin{equation}\label{eq:35W11}
 W^{(3,5)}_{1,1}(p)=
 -\frac{p^6+3p^4-9p^3+36p^2+54p-54}{21(p^2+3)^4}.
\end{equation}
Its lack of a definite parity in $p$ reflects the nonzero leading zeroth-order coefficient field $a_0$. The general three-boundary formula becomes
\begin{equation}\label{eq:35W03}
 W^{(3,5)}_{0,3}(p_1,p_2,p_3)=
 \sum_{\alpha=\pm\ii\sqrt3}\frac{3}{2(\alpha+2)}
           \prod_{i=1}^3\frac1{(p_i-\alpha)^2}.
\end{equation}
The sum over complex roots combines into a rational function over $\mathbb Q$, with external denominator $\prod_i(p_i^2+3)^2$.

The matrix construction also gives the full genus-one two-boundary function. Both one-boundary functions through genus two and the full two-variable $W_{1,2}$ agree with independent EO calculations. Appendix~\ref{app:NUdata} gives $W_{2,1}$ and a compact specialization of $W_{1,2}$. Here, as in the unitary examples, comparison with the literature identifies the operators and string equations, while the independent residues test the explicit rational amplitudes.

The dual $(5,3)$ projection has $e^\vee=\pi$ and $y^\vee=-e/405$. Its monic ramification polynomial is
\begin{equation}\label{eq:53ram}
 \mathscr D_{53}(p)=p^4+9p^2+6p+9.
\end{equation}
We perform the fifth-order reconstruction with four independent quantum fields. Its result agrees through $\h^4$ with
\begin{align}\label{eq:53fields}
 u_3^\vee&=\frac53b,&u_2^\vee&=\frac53a,\nonumber\\
 u_1^\vee&=\frac59(b^2-\h^2b''),&
 u_0^\vee&=\frac{10}{9}ab+\frac{5\h^2}{18}a''.
\end{align}
Restoring the ordering derivative terms identifies the complete fifth-order operator with the Baker--Akhiezer partner in the literature. Its independently constructed third-order partner is the original $Q$.

The quotient algebra \eqref{eq:quotient} lets us work with the quartic \eqref{eq:53ram} exactly. The resulting $W_{1,1}$, $W_{2,1}$ and full $W_{1,2}$ agree with EO as rational functions. Their poles and numerators differ from the third-order answers, as expected from the changed projection. We have therefore checked both the common string solution and the different boundary observables it defines.

The Yang--Lee model provides a complementary connection back to KdV. In the fifth-order realization choose
\begin{equation}\label{eq:52curve}
 e(z)=z^5-5z^3+\frac{15}{2}z,\qquad
 \pi(z)=z^2-2,\qquad \kappa=5,\qquad U=x^{1/3}.
\end{equation}
The disk is $y=-(z^2-2)/5$. The coefficient $15/2$ in $e$ fixes the homogeneous Yang--Lee background used here. Replacing it by the Chebyshev coefficient $5$ would change the leading bracket and hence the reconstruction problem.

All four fields are left free. The two fields that would break the skew-adjoint reduction vanish through genus two, and the complete operator becomes
\begin{align}\label{eq:52Q}
 Q={}&D^5-5\mathfrak uD^3-\frac{15}{2}\h\mathfrak u'D^2
 +\left(\frac{15}{2}\mathfrak u^2-\frac{25}{4}\h^2\mathfrak u''\right)D\nonumber\\
 &+\frac{15}{2}\h\mathfrak u\mathfrak u'
                    -\frac{15}{8}\h^3\mathfrak u''',\nonumber\\
 \mathfrak u={}&U-\frac{\h^2}{36}U^{-6}
                    -\frac{7\h^4}{432}U^{-13}+O(\h^6).
\end{align}
Its partner is $P=D^2-2\mathfrak u$. In symmetric fields, $u_3=-5\mathfrak u$ and $u_1=15\mathfrak u^2/2+5\h^2\mathfrak u''/4$, while $u_2=u_0=0$ at the calculated orders.

The normalization map in Appendix~\ref{app:published} identifies \eqref{eq:52Q} with the Yang--Lee operator in the literature \cite[Appendix~B.3]{ChanIrieNiednerYeh2016}. Its string equation reads
\begin{equation}\label{eq:YLstring}
 \mathfrak u^3-\frac{\h^2}{2}\mathfrak u\mathfrak u''
 -\frac{\h^2}{4}(\mathfrak u')^2+\frac{\h^4}{40}\mathfrak u''''=x.
\end{equation}
The series in \eqref{eq:52Q} satisfies this equation through $\h^4$, and the independent commutator gives $[P,Q]=5\h$ through $\h^5$. Reconstructing from the second-order projection returns the same $\mathfrak u$ with the operators exchanged. This brings the higher-order calculation back to a familiar KdV minimal-string equation and provides another comparison with the literature.

The fifth-order boundary amplitude remains a different observable. At genus one,
\begin{equation}\label{eq:52W11}
 W^{(5,2)}_{1,1}(p)=-\frac{B_{52}(p)}{36(2p^4-6p^2+3)^4}.
\end{equation}
The numerator $B_{52}$ is given in Appendix~\ref{app:NUdata}.
For the second-order projection,
\begin{equation}\label{eq:25W11}
 W^{(2,5)}_{1,1}(p)=-\frac{2p^2+3}{72p^4}.
\end{equation}
The difference is expected: the first projection has four finite simple ramification points, while the second has one. In both projections the full genus-two one-boundary function and the full genus-one two-boundary function pass independent EO comparisons. The corresponding genus-two polynomials are included in Appendix~\ref{app:NUdata}.

\section{Higher-Airy limits and intersection theory}\label{sec:airy}

All the preceding inverse calculations relied on simple ramification. What remains of the construction when branch points meet? The higher-Airy family gives a particularly transparent setting for this question:
\begin{equation}\label{eq:Airy}
 Q=D^r-\lambda x,\qquad p^r=E+\lambda x,
 \qquad E(p,0)=p^r,\qquad y(p)=\frac p\lambda,
\end{equation}
with $\lambda\ne0$. For $r>2$, the finite ramification point is multiple and the inverse used in \eqref{eq:descent} becomes singular. This locus lies outside the regular domain of the uniqueness theorem. The scalar ODE and its concomitant remain well defined there, and the derivative
\begin{equation}
 \cD=\frac{\lambda}{rp^{r-1}}\partial_p
\end{equation}
turns the coefficient calculation into a recursion on monomials. Here the operator is already specified, so the calculation tests the resolvent and its normalization on a degenerate curve.

There is a useful geometric description of these amplitudes. Their coefficients generate $r$-spin intersection numbers, which are integrals of characteristic classes on moduli spaces of curves equipped with $r$-spin data \cite{Witten1993,FaberShadrinZvonkine2010,BertolaDubrovinYang2015,BertolaDubrovinYang2016,EynardEtAl2023}. At $r=2$ this reduces to Witten--Kontsevich theory \cite{Witten1991,Kontsevich1992}. Each marked point carries a descendant index $d_i\geq0$ and a spin-sector label $\nu_i$, with $0\leq\nu_i\leq r-2$. To state the dictionary, write
\begin{equation}\label{eq:multifactorial}
 m_{!!}^{(r)}=m(m-r)(m-2r)\cdots
\end{equation}
with the product ending at the last positive factor. The stable normalization dictionary is
\begin{align}\label{eq:rspindictionary}
 W^{\mathrm A,r;\lambda}_{g,n}(\mathbf p)
 ={}&\lambda^{2g-2+n}\sum_{\mathbf d,\boldsymbol\nu}
 (-r)^{g-1-\sum_i d_i}
 \left\langle\prod_{i=1}^n\tau_{d_i,\nu_i}\right\rangle_g
 \prod_{i=1}^n\frac{(rd_i+\nu_i+1)_{!!}^{(r)}}
 {p_i^{rd_i+\nu_i+2}},\\
 &\sum_i(rd_i+\nu_i+1)=(r+1)(2g-2+n).\nonumber
\end{align}
The superscript $\mathrm A$ labels the Airy family within the type-$A$ construction. Equation~\eqref{eq:rspindictionary} is the higher-Airy determinantal formula in the curve variables $E_i=p_i^r$ \cite[Theorem~5.1, Eq.~(5.19)]{EynardEtAl2023}. The disk rescaling fixes the factor of $\lambda$ at every genus.

For one marked point the dimension constraint reads
\begin{equation}\label{eq:onepointselection}
 rd+\nu=2g(r+1)-(r+2),\qquad \nu=(2g-2)\bmod r.
\end{equation}
The constraint tells us which coefficient can occur before we do the ODE calculation. If it requires $\nu=r-1$, the label lies outside the allowed narrow sector and the contribution vanishes. The five-Airy genus-three answer below illustrates this rule.

The closed one-point intersection formula of Liu, Vakil and Xu \cite[Proposition~5.4]{LiuVakilXu2011}, combined with \eqref{eq:rspindictionary}, gives
\begin{align}\label{eq:Airyallr}
 W^{\mathrm A,r;\lambda}_{1,1}(p)
 &=-\frac{\lambda(r^2-1)}{24r\,p^{r+2}},\\
 W^{\mathrm A,r;\lambda}_{2,1}(p)
 &=\frac{\lambda^3(r^2-1)(r^2-9)(2r+1)(2r+3)}
 {640r^3\,p^{3r+4}}.\nonumber
\end{align}
The intersection formula gives these expressions directly at arbitrary scalar order. We can therefore compare the scalar calculation with the literature using the same normalization at every scalar order. The factor $r^2-9$ predicts the genus-two zero at $r=3$, while at $r=2$ the answer has the negative sign of the monic second-order convention.

With $\lambda=1$, the first fourth- and fifth-order values are
\begin{align}
 W^{\mathrm A,4}_{1,1}&=-\frac5{32p^6},&
 W^{\mathrm A,4}_{2,1}&=\frac{2079}{8192p^{16}},&
 W^{\mathrm A,4}_{3,1}&=\frac{3132675}{1048576p^{26}},\label{eq:Airy4}\\
 W^{\mathrm A,5}_{1,1}&=-\frac1{5p^7},&
 W^{\mathrm A,5}_{2,1}&=\frac{429}{625p^{19}},&
 W^{\mathrm A,5}_{3,1}&=0.\label{eq:Airy5}
\end{align}
The sign of each amplitude follows from the full conversion between intersection numbers and curve coordinates. The conversion contains both the power $(-r)^{g-1-d}$ and the multifactorial. For example, at $r=4,g=2$ one has $(d,\nu)=(3,2)$ and
\begin{equation}
 \frac{2079}{8192}=(-4)^{-2}(15\cdot11\cdot7\cdot3)
                 \frac3{2560}.
\end{equation}
The numerical factor and the sign therefore follow from the same conversion.

The direct fourth- and fifth-order calculations continue through genus six. After the normalization in \eqref{eq:rspindictionary}, every one-point coefficient agrees with the closed formula of Liu--Vakil--Xu \cite{LiuVakilXu2011}; the exact values and the forbidden-sector zero are collected in Table~\ref{tab:Airy} of Appendix~\ref{app:airydata}.

With additional boundaries, the comparison also tests how the descendant labels are distributed among the marked points. At orders four and five, the complete $W_{1,2}$ agrees with the string, dilaton and genus-one gluing relations; the planar three- and four-boundary functions agree with the corresponding $r$-spin rules. Appendix~\ref{app:airydata} gives the calculation, together with the specified subsets of coefficients checked for the larger $W_{2,2}$ and $W_{3,2}$ outputs.

\section{Discussion}\label{sec:discussion}
We began by asking how much of a quantum differential equation is already encoded in its classical curve. The examples show why the background motion must be part of that question. Pure gravity and Ising illustrate this dependence: they can share an endpoint energy polynomial, with their different motion in $x$ selecting different disks and amplitudes. Once this motion is specified, the analytic behavior of the one-boundary integral carries enough information to constrain the quantum fields very strongly.

The reason is visible near a branch point. The lower-genus data fix the most singular terms in the resolvent at genus $g$. Matching these terms determines a candidate primitive, after which the remaining equations leave at most one choice for the fields. This is the content of the uniqueness theorem on the regular domain. The final remainder also has a useful meaning: it tells us whether the candidate actually solves the problem in the chosen analytic class. A nonzero remainder is an obstruction to reconstruction in this class. When an admissible string solution is available, the exact trace identity supplies a primitive and the uniqueness theorem identifies it with the one obtained by reconstruction.

This separates two questions that are often answered together. A string equation selects a background and provides its quantum completion. The inverse calculation asks how much of that completion is forced by the analytic form of an observable. In the examples considered here, comparison with the literature and independent recursion calculations show that the primitive criterion recovers the expected completion. The general theorem explains its uniqueness, while existence still requires either compatibility at each genus or an admissible string solution. Characterizing the leading families for which compatibility holds at every genus would therefore be a substantial next step. In particular, it would be useful to learn when a string partner can be recovered from the leading data themselves.

The additional boundaries reveal information that is less visible in the scalar answer. The concomitant gives a rank-one projector, and cyclic traces of this projector retain the full field dependence needed for connected amplitudes. In Ising, the mixed term in $W_{1,2}$ connects the two ramification points. It shows concretely why the leading singularity at each branch point is only part of a global two-boundary function. The identification of the cyclic traces with topological-recursion amplitudes depends on properties of the full differential system. Similarly, exchanging the two operators of a string pair changes the energy projection and the disk, so it changes the marked-boundary observable even when the same background solution is being described.

Each extension of the basic examples has a purpose. The higher-order unitary calculations show that the method can determine several homogeneous fields at once. They give a particularly simple genus-one formula for the highest field. The fourth-order example in Appendix~\ref{app:diagnostics} explains how the lower fields acquire further contributions that require their own calculation. The nonunitary models show that the polynomial-remainder construction works just as well when the ramification points are complex. Together, these examples bring the role of the regularity conditions into focus: simple ramification and nonzero branch-value velocities support the reconstruction in both unitary and nonunitary models, including curves with complex roots.

When ramification points collide, or their branch values stop moving, the multiplication map used in the proof becomes singular. The higher-Airy calculations show that the scalar ODE can remain meaningful there, and its comparison with $r$-spin theory provides a useful guide. An inverse reconstruction theorem on this locus remains an open problem. A reformulation of the allowed primitives, informed by higher-ramification recursion, may make it possible to ask the same reconstruction question under these different local conditions \cite{BertolaDubrovinYang2015,BouchardEynard2013}.

There is also a question beyond the formal expansion considered here. Our construction determines coefficients in powers of $\h$; choosing contours, Stokes data and a nonperturbative completion requires further analytic information. Johnson and Rodrigues have extended the two-ODE approach to transseries and large-genus behavior in the second-order setting \cite{JohnsonRodrigues2026}, while resurgent determinantal methods give asymptotics of intersection numbers \cite{EynardEtAl2023}. Extending these connections to higher-order reconstruction would help clarify which parts of a nonperturbative theory can be read from the observable and which must be specified independently. The present result settles the perturbative uniqueness question on the regular domain, and gives a finite way to examine compatibility genus by genus.

\section*{Acknowledgments}
\addcontentsline{toc}{section}{Acknowledgments}
This article is our first attempt at AI-assisted research where we prompted some simple ideas and tested them on many examples which led to the development of this work. Claude Opus 5.0 and GPT Astra 6.0 were used in developing symbolic calculations and consistency checks for our computed examples. The authors have checked the results presented here, with the scope and remaining limitations of independent verification stated beside the relevant formulas, and take responsibility for the scientific content.

During the final preparation of this manuscript, we became aware of related independent work by Clifford V. Johnson.
W.A. is supported by the U.S. Department of Energy under grant DE-SC 0011702. J.N. is supported by the Beijing Institute of Mathematical Sciences and Applications.

\appendix
\addtocontents{toc}{\protect\setcounter{tocdepth}{1}}

\section{The reconstruction and trace statements}\label{app:proofs}
\subsection{The multiplication map and uniqueness}
We state the regularity assumptions precisely because the inverse calculation depends on them. Let $\mathbb K$ be a differential field of characteristic zero containing the leading fields and their derivatives. Work on a local domain where $\Delta\in\mathbb K[p]$ is square-free and $V$ is nonzero at every zero of $\Delta$. Equivalently,
\begin{equation}
 \gcd(\Delta,\Delta_p)=1,\qquad \gcd(\Delta,V)=1.
\end{equation}
These are algebraic versions of distinct ramification points and nonzero branch-value velocities. They apply to real and complex roots alike.

\begin{lemma}
For $k\geq1$, multiplication by $kV\Delta_p$ is invertible in $\mathbb K[p]/(\Delta)$. Its determinant in any polynomial basis is \eqref{eq:determinant}.
\end{lemma}
\begin{proof}
Extend the coefficient field temporarily so that all roots $\alpha_i$ are available. Evaluation at those distinct roots identifies the quotient algebra with a product of $r-1$ copies of that field. Multiplication by $kV\Delta_p$ becomes diagonal, with entries $kv_i\Delta_p(\alpha_i)$. All are nonzero, and their product gives the determinant. The inverse descends to the original coefficient field: B\'ezout's identity constructs it directly from $\Delta$ and $V\Delta_p$ by the extended Euclidean algorithm. The result is independent of the ordering of the roots.
\end{proof}

\begin{theorem}[Uniqueness at a fixed genus]\label{thm:unique}
Fix the leading fields and all corrections below genus $g\geq1$. On the regular domain, at most one field polynomial $F_g$, with $\deg F_g<r-1$, and one primitive in the class \eqref{eq:slots} satisfy
\begin{equation}
 \cD\Qhat_g=S_g+
 \frac{F_g\Delta_p-(\partial_pF_g)\Delta}{\Delta^3}.
\end{equation}
Consequently an admissible reconstruction, when it exists, is unique at each genus for every scalar order $r$.
\end{theorem}
\begin{proof}
Suppose that two solutions exist, and subtract their equations. Write the primitive difference as $q$ and the field difference as $F$. If $q\ne0$, let $k\geq2$ be its highest occupied slot in the unique reduced expansion \eqref{eq:slots}. The coefficient at denominator power $k+2$ in $\cD q$ is $[kV\Delta_pA_k]$, which is nonzero by the lemma. At a root where this class is nonzero, the derivative has a pole of order $k+2\geq4$. The right-hand side, $(F\Delta_p-F_p\Delta)/\Delta^3$, has poles of order at most three. This is a contradiction, so $q=0$.

It follows that $F\Delta_p-F_p\Delta=0$, or $\partial_p(F/\Delta)=0$. Thus $F=c(x)\Delta$. Since $\deg F<\deg\Delta$, one has $F=0$. The argument applies successively at each genus, once lower-genus coefficients have been fixed.
\end{proof}

Admissibility also fixes the integration constant. Two primitives of the same resolvent differ by a function with zero fixed-energy derivative. The strongest-pole argument shows that every nonzero admissible rational function has a nonzero derivative, so the two primitives coincide. The pole and infinity conditions thus fix the freedom to add a function of $E$.

The uniqueness argument needs only a finite admissible primitive; it does not assume a numerical cutoff on its highest slot. To obtain the prescribed finite search space used by the algorithm, we additionally need a bound on the poles produced by the scalar recurrence. The following induction supplies that bound.

\begin{localpolelemma}[Local pole bounds]
Fix a simple ramification point and work on a background domain where the coefficient fields and the finitely many jets at the perturbative order under consideration are regular. Let $\operatorname{pole}_\alpha f$ denote the pole order of $f$ in $p-\alpha(x)$, with value zero for a regular function. For every $m\geq1$, the formal scalar recurrence and the concomitant obey
\begin{equation*}
 \operatorname{pole}_\alpha w_m\leq3m-1,\qquad
 \operatorname{pole}_\alpha\Sigma_m\leq3m-1,\qquad
 \operatorname{pole}_\alpha\Rhat_m\leq3m+1.
\end{equation*}
The paired coefficients $\widetilde w_m$ satisfy the same bound as $w_m$.
\end{localpolelemma}
\begin{proof}
Put $t=p-\alpha(x)$. Since $\Delta$ has a simple zero, the operator $\cD=\partial_x|_p-(V/\Delta)\partial_p$ raises a pole order by at most two. Its action on the regular leading momentum is more specific: $\cD p=-V/\Delta$ has order at most one. Hence
\begin{equation*}
 \operatorname{pole}_\alpha(\cD^k p)\leq2k-1,
 \qquad k\geq1.
\end{equation*}
The coefficients of each $c_j$ and their background derivatives are regular in the curve variable and introduce no additional $p$-poles.

Every Bell-polynomial term is a product of factors $\cD^k w_\ell$, with an explicit power $\h^k$ for a factor differentiated $k$ times. Such a factor has total perturbative order $\ell+k$. Assume inductively that the bound for $w_\ell$ is known for $1\leq\ell<m$. A factor with $\ell+k>0$ then has pole order at most $3\ell+2k-1$; the case $\ell=0$, $k\geq1$ follows from the preceding bound on $\cD^k p$. Undifferentiated factors $w_0=p$ are regular. If a monomial has $q\geq1$ positive-order factors, $K$ derivatives distributed among them, and perturbative order $M$ from those factors, its pole order is at most $3M-K-q$.

At order $\h^m$ in the scalar equation, the only terms containing a single undifferentiated factor $w_m$ combine to $\Delta w_m$. Remove that term to form the known source in \eqref{eq:WKBrec}. Every remaining term of order $m$ either has at least two positive-order factors, has a differentiated factor, or obtains a positive power of $\h$ from a coefficient $c_j$. In the first two cases $K+q\geq2$, and in the last case $M<m$ (or the term has no singular factor). The known source therefore has pole order at most $3m-2$. Division by the simple zero of $\Delta$ proves the bound $3m-1$ for $w_m$, including the initial case $m=1$. The paired series has the same estimate by \eqref{eq:parity}.

Apply the same product counting to the finite concomitant sum. At positive order $m$, any singular term contains at least one positive-order factor, so its pole order is at most $3m-1$. This proves the bound for $\Sigma_m$. Finally, $\Rhat_0=\Delta^{-1}$ has pole order one. Inducting in the inverse-series recurrence \eqref{eq:Rrec}, every summand obeys
\begin{equation*}
 \operatorname{pole}_\alpha
 \left(\frac{\Sigma_k\Rhat_{m-k}}{\Delta}\right)
 \leq (3k-1)+\bigl(3(m-k)+1\bigr)+1=3m+1,
\end{equation*}
where the bound for $m-k=0$ is precisely the leading pole order. This gives the claimed resolvent estimate.
\end{proof}

At even order $m=2g$ the bound is $6g+1$. If an admissible primitive had a highest occupied slot $k>6g-1$, the invertible multiplication map of the first lemma would force its derivative to have a nonzero pole of order $k+2>6g+1$. Thus the allowed search cutoff is $K_g^{\max}=6g-1$. The bound gives a finite search space, while the residual equations still decide compatibility. Special backgrounds may have smaller pole orders through cancellations.

\subsection{Why a string pair supplies a primitive}
The exact identity \eqref{eq:traceprimitive} follows directly from compatibility. Using $\h\partial_x\cM=[L,\cM]$ and cyclicity of the trace,
\begin{align}
 \partial_x\tr(\mathsf B\cM)
 &=\tr\left[\left(\partial_x\mathsf B+
                   \frac{[\mathsf B,L]}\h\right)\cM\right]\\
 &=\tr\bigl((\partial_E L)\cM\bigr)
   =\tr(e_{r1}\cM)=\cM_{1r}=\Rhat.
\end{align}
Here the first line uses only the spatial projector equation, and the second uses the zero-curvature relation \eqref{eq:flatness}. The matrix $e_{r1}$ is the derivative of the companion matrix with respect to its spectral parameter. This explains why the particular scalar entry selected by the concomitant appears.

There is also a scalar expression for the same trace. Define
\begin{equation}
 T_0=\frac{P\psi}{\psi},\qquad
 T_{j+1}=wT_j+\h\cD T_j.
\end{equation}
Then $T_j\psi=D^jP\psi$. In terms of the row coefficients $\ell_j$ in \eqref{eq:rowcolumn}, the trace is
\begin{equation}\label{eq:scalartrace}
 \tr(\mathsf B\cM)
 =-\frac{\Rhat}{\kappa}\sum_{j=0}^{r-1}\ell_j T_j.
\end{equation}
Thus the trace check uses the same formal scalar WKB coefficients that already determine the resolvent.

\begin{proposition}[Transfer from an admissible string solution]\label{prop:transfer}
Let a specified formal pair $[P,Q]=\kappa\h$ have the leading fields used in the inverse problem and the symmetric, even-field convention of Section~\ref{sec:resolvent}. Suppose that, at every positive genus, $[\h^{2g}]\tr(\mathsf B\cM)$ belongs to the class \eqref{eq:slots}. On the regular domain, the inverse reconstruction exists and agrees with this string solution order by order.
\end{proposition}
\begin{proof}
The trace identity supplies a solution of the primitive equation. By hypothesis it is admissible. Theorem~\ref{thm:unique} identifies it with the unique inverse solution, first at genus one and then inductively at every higher genus.
\end{proof}

The proposition makes clear what must be checked for a chosen string solution: rationality, the allowed finite poles and regularity at infinity. Once these hold, the trace identity and uniqueness theorem give the inverse result. Identifying the higher cyclic traces with EO differentials requires the additional topological-type properties of the differential system \cite{BergereBorotEynard2015}.

\subsection{A genus-one source at arbitrary scalar order}
At genus one there is a compact way to obtain the scalar source using the symbol of the operator. It provides a useful check on the recurrence of Section~\ref{sec:resolvent}. All derivatives in this subsection are at fixed $p$, except for the explicitly displayed energy derivative. Put
\begin{equation}
 V=E_x,\qquad J=E_{xx},\qquad
 \partial_E=\Delta^{-1}\partial_p,
 \qquad F_1=\sum_{j=0}^{r-2}u_{j,1}p^j.
\end{equation}
For the anticommutator ordering used here, the Weyl symbol of the operator is
\begin{equation}
 Q_{\rm W}=E+\h^2\left(F_1+\frac18 E_{xxpp}\right)+O(\h^4).
\end{equation}
Converting from anticommutator ordering to Weyl ordering introduces fixed derivative terms for general powers of $D$. For example, $\{u,D^2\}/2$ has Weyl symbol $up^2+\h^2u''/4$ in the convention $D=\h\partial_x$. The derivative term above is fixed by this ordering conversion.

For an external spectral parameter $\epsilon$, set $t=E(p,x)-\epsilon$. The inverse symbol is defined by Weyl composition, rather than ordinary pointwise multiplication. In the convention $D=\h\partial_x$, that composition is
\begin{equation*}
 f\star g
 =f\exp\!\left[\frac\h2
 \left(\overleftarrow{\partial_p}\overrightarrow{\partial_x}
       -\overleftarrow{\partial_x}\overrightarrow{\partial_p}\right)\right]g.
\end{equation*}
Let $G_{\rm W}$ be the unique formal symbol satisfying
$(Q_{\rm W}-\epsilon)\star G_{\rm W}=G_{\rm W}\star(Q_{\rm W}-\epsilon)=1$
with leading term $t^{-1}$. Expanding this inverse equation gives
\begin{align}
 (Q-\epsilon)^{-1}_{\rm W}
 ={}&\frac1t-\frac{\h^2(F_1+E_{xxpp}/8)}{t^2}
       +\frac{\h^2(E_{pp}J-E_{xp}^2)}{4t^3}\\
 &-\frac{\h^2(E_{pp}V^2-2E_{xp}V\Delta+J\Delta^2)}{4t^4}
       +O(\h^4).\nonumber
\end{align}
To relate this inverse to the scalar pairing, first work away from ramification and select a simple root $p_a(x,\epsilon)$ of $E(p,x)=\epsilon$. The label $a$ here selects a sheet. In the formal kernel of a Weyl operator, setting the two position arguments equal leaves the momentum integral of its symbol. Selecting this root by a positively oriented local contour therefore gives the normalized sheet contribution
\begin{equation*}
 \Rhat_a(x,\epsilon;\h)
 =\Res_{\xi=p_a(x,\epsilon)}
       G_{\rm W}(x,\xi;\epsilon,\h)\,\dd\xi.
\end{equation*}
The variable $\xi$ is the integration momentum, and the residue is taken coefficient by coefficient about the classical root.

The normalization and the identification with the concomitant can be seen without stopping at the leading term. For $(Q-\epsilon)f=j$, the companion vector $\boldsymbol f=(f,Df,\ldots,D^{r-1}f)^t$ satisfies
\begin{equation*}
 \h\partial_x\boldsymbol f=L\boldsymbol f+e_r j,
\end{equation*}
where $e_r$ is the last coordinate vector. Variation of constants gives the mode contribution to its scalar Green kernel as
\begin{equation*}
 G_a(x,x';\epsilon)
 =\frac1\h e_1^t\Psi(x,\epsilon)\Pi_a
                    \Psi^{-1}(x',\epsilon)e_r.
\end{equation*}
Here $\Pi_a$ selects the same formal sheet, and the oriented source-jump convention agrees with the positive residue above. On the diagonal, $\h G_a(x,x;\epsilon)=\cM_{a,1r}=\psi_a\chi_a$, with the concomitant normalized to $\mathcal B=\h$. The factor $1/\h$ also appears in the Weyl kernel measure $\dd\xi/(2\pi\ii\h)$. Thus the Weyl parametrix and the paired solutions construct the same normalized local inverse, order by order, rather than merely sharing the leading value $1/\Delta$. This is a formal identification on a separated sheet; actual global Green functions still require boundary or contour data, as discussed in Section~\ref{sec:resolvent} \cite{GelfandDikii1975,BergereBorotEynard2015}.

For a function $f(x,p)$ regular near the selected simple root, changing the residue coordinate from $p$ to $E$ gives
\begin{equation*}
 \Res_{p=p_a(x,\epsilon)}
 \frac{f(x,p)\,\dd p}{(E(p,x)-\epsilon)^{k+1}}
 =\left.\frac1{k!}\partial_E^k\left(\frac f\Delta\right)
   \right|_{E=\epsilon},
 \qquad \partial_E=\Delta^{-1}\partial_p.
\end{equation*}
Applying this identity to the displayed inverse symbol and suppressing the sheet label yields
\begin{align}\label{eq:genusonesource}
 \Rhat_2={}&-\partial_E\left(\frac{F_1+E_{xxpp}/8}{\Delta}\right)
  +\frac18\partial_E^2\left(\frac{E_{pp}J-E_{xp}^2}{\Delta}\right)\nonumber\\
 &-\frac1{24}\partial_E^3\left(
     \frac{E_{pp}V^2-2E_{xp}V\Delta+J\Delta^2}{\Delta}\right).
\end{align}
Setting $F_1=0$ gives the source $S_1$ in \eqref{eq:fieldvariation}. The term containing $F_1$ is exactly the universal new-field contribution in that formula. At $r=2$, with $E=p^2+u(x)$, the result reduces to
\begin{equation}
 \Rhat_2=\frac{F_1}{4p^3}+\frac{u''}{16p^5}
                              +\frac{5(u')^2}{64p^7}.
\end{equation}
The sign follows our monic convention; conversion to the Schr\"odinger operator in \eqref{eq:GD} requires the corresponding sign change. At the moving sixth- and seventh-order backgrounds the compact source and a separately organized Bell/concomitant calculation agree exactly.

\subsection{Local involutions for the residue check}
Near a simple ramification point $\alpha$, put $p=\alpha+t$ and $X_j=E^{(j)}(\alpha)$. The nontrivial local solution of $E(\sigma(p))=E(p)$ has the expansion
\begin{equation}
 \sigma(\alpha+t)=\alpha-t+s_2t^2+s_3t^3+s_4t^4+\cdots,
\end{equation}
where coefficient matching gives
\begin{align}
 s_2&=-\frac{X_3}{3X_2},\qquad
 s_3=-\frac{X_3^2}{9X_2^2},\\
 s_4&=-\frac{9X_2^2X_5-30X_2X_3X_4+40X_3^3}{540X_2^3}.
\end{align}
The condition $X_2\ne0$ is simple ramification. Set $Y_j=y^{(j)}(\alpha)$ and assume $y$ is regular there with $Y_1\ne0$, as required in Section~\ref{sec:EO}. Substituting the involution into the recursion kernel and Bergman differential, including $\dd\sigma=\sigma'(p)\dd p$, gives the complete genus-one answer
\begin{align}\label{eq:EO11}
 W_{1,1}(p)=\sum_\alpha\bigg[&
 -\frac1{8X_2Y_1(p-\alpha)^4}
 +\frac{X_3}{24X_2^2Y_1(p-\alpha)^3}\nonumber\\
 &+\frac{X_4/X_2-X_3^2/X_2^2-X_3Y_2/(X_2Y_1)+Y_3/Y_1}
 {48X_2Y_1(p-\alpha)^2}\bigg].
\end{align}
This sum retains the lower local poles as well as the leading one. At a fixed topology only finitely many involution coefficients are needed. If the roots of $E_p$ are algebraic, the residue coefficients can be reduced modulo the ramification polynomial before their trace is taken.

\section{Background selection and the limits of a compact formula}\label{app:diagnostics}
\subsection{Stationary leading equations}
The hierarchy offers a systematic way to specify leading fields before attempting quantum reconstruction. With
\begin{equation}
 \lambda(p)=p^r+\sum_{j=0}^{r-2}u_jp^j,
\end{equation}
define the dispersionless densities
\begin{equation}\label{eq:stationarydensities}
 h_{\beta,n}=\frac{\Gamma(\beta/r)}{\Gamma(n+1+\beta/r)}
                 [p^{-1}]\lambda^{n+\beta/r},
 \qquad \beta=1,\ldots,r-1,\quad n\geq0.
\end{equation}
Here the $u_j$ are leading fields and $[p^{-1}]$ extracts the coefficient of $p^{-1}$ in the Laurent series. If $k=rn+\beta$ and $q=1/p$, the required coefficient is
\begin{equation}
 [q^{k+1}]\left(1+\sum_j u_jq^{r-j}\right)^{k/r}.
\end{equation}
Only a finite expansion is needed for a prescribed density. A useful source choice is
\begin{equation}
 \mathcal S=xu_0+\sum_{\beta,n}t^{\beta,n}h_{\beta,n},
 \qquad \frac{\partial\mathcal S}{\partial u_j}=0.
\end{equation}
The coefficients $t^{\beta,n}$ are fixed background parameters in this normalization of the stationary problem. Choosing the source $xu_0$ selects this stationary problem within the hierarchy \cite{DrinfeldSokolov1985,Dubrovin1996,BelavinDubrovinMukhametzhanov2014}.

At second order, $h_{1,n}=u_0^{n+1}/(n+1)!$, so the stationary equation has the usual dispersionless KdV form. At third order, writing $u_1=b$ and $u_0=a$, the first densities are
\begin{equation}
 \begin{array}{c|ccc}
 & n=0& n=1&n=2\\\hline
 \beta=1&b&ab&b(18a^2-b^3)/36\\[2pt]
 \beta=2&a&a^2/2-b^3/18&a(3a^2-b^3)/18
 \end{array}
\end{equation}
These make the source dependence visible with little algebra. For Ising, $t^{2,2}=-2/3$ and the branch $a=0$ give $b^3=-27x$. The choice $t^{1,0}=\tau$, $t^{2,1}=c$ and $t^{1,1}=\sigma$ gives
\begin{equation}
 ca+\sigma b+x=0,\qquad
 \sigma a-\frac c6b^2+\tau=0.
\end{equation}
On its regular domain,
\begin{equation}
 a'=-\frac{cb}{c^2b+3\sigma^2},\qquad
 b'=-\frac{3\sigma}{c^2b+3\sigma^2}.
\end{equation}
Substitution into \eqref{eq:genericB}--\eqref{eq:genericA} yields
\begin{equation}
 A=-\frac{9c^3\sigma^3}{4(c^2b+3\sigma^2)^4},\qquad
 B=-\frac{81\sigma^4(c^2b+\sigma^2)}{8b^2(c^2b+3\sigma^2)^4}.
\end{equation}
This illustrates the reconstruction formula for general fields. Identifying a minimal-model background at these times would require an additional independent comparison.

Changing the source reaches a different homogeneous family. The action
\begin{equation}
 \mathcal S=xb+T\frac{a(3a^2-b^3)}{18}
\end{equation}
gives $a^2=b^3/9$ and $ab^2=6x/T$. For the normalization of the $(3,5)$ example in Section~\ref{sec:nonunitary}, $T=2/243$ selects $a=9x^{3/7}$ and $b=9x^{2/7}$. Its classical fractional-power partner has constant bracket $\kappa=-405$. The disk coefficient $1/405$ is therefore fixed before any positive-genus comparison.

The density coordinates are related to the operator coefficients by a change of variables that becomes nonlinear at higher scalar order. Already at $r=4$,
\begin{equation}
 h_{1,0}=u_2,\qquad h_{2,0}=u_1,\qquad
 h_{3,0}=u_0-\frac18u_2^2.
\end{equation}
Thus a formula stated in flat density coordinates requires a change of variables before it is compared with the symmetric coefficient fields used here. The comparison therefore requires both this explicit change of variables and an identification of the physical background, alongside the scaling weights.

\subsection{A restricted genus-two calculation}
The genus-two calculation can be carried out with $a(x)$ arbitrary and the leading $b$ retained as a symbolic nonzero constant. Work locally where $a_1\ne0$, with $a_j=\partial_x^j a$. At genus one, equations~\eqref{eq:genericA}--\eqref{eq:genericB} give
\begin{equation*}
 u_{0,1}=0,\qquad
 u_{1,1}=\frac{a_1a_3-a_2^2}{4a_1^2}.
\end{equation*}
Only the leading $b$ is constant: this nonconstant genus-one correction and all of its ordering terms must be retained in the genus-two scalar source. Applying the descending reconstruction gives
\begin{equation}
 u_{1,2}=0,\qquad u_{0,2}=-\frac{b\mathcal N}{480a_1^6},
\end{equation}
where
\begin{align}
 \mathcal N={}&5a_1^4a_6-41a_1^3a_2a_5-73a_1^3a_3a_4
                 +204a_1^2a_2^2a_4\\
 &+285a_1^2a_2a_3^2-700a_1a_2^3a_3+320a_2^5.\nonumber
\end{align}
This statement is a rational identity in $b$ and the independent jets of $a$, rather than an interpolation from selected values of $b$. To check it, use
\begin{equation*}
 \Delta=3p^2+b,\qquad
 \cD=\sum_{j\geq0}a_{j+1}\frac{\partial}{\partial a_j}
             -\frac{a_1}{\Delta}\partial_p,
\end{equation*}
where only finitely many jet derivatives enter at any fixed order. Form $S_2$ with the genus-one fields above and use the slots $2,\ldots,11$ in \eqref{eq:descent}. In the resulting primitive $\Qhat_2$, the two residual checks are
\begin{equation*}
 \cD\Qhat_2-S_2-\frac{6p\,u_{0,2}}{\Delta^3}=0,
 \qquad [p]A_2=0.
\end{equation*}
Both vanish with $b$ symbolic; the polynomial $A_2$ here is the denominator-power-two numerator in \eqref{eq:slots}, not a field correction. The regularity assumptions reduce to $b\ne0$ and $a_1\ne0$, so Theorem~\ref{thm:unique} fixes the displayed genus-two fields in this constant-leading-$b$ family. No claim about arbitrary moving $b(x)$ or all-genus existence follows from this finite-order calculation. In the linear specialization $a=-x$, both genus-two field corrections vanish. The one-boundary function is
\begin{equation}
 W_{2,1}(p;b)=\frac{315bp(99p^4-42bp^2+2b^2)}{(3p^2+b)^{10}}.
\end{equation}
The $b\to0$ limit vanishes, consistently with the forbidden narrow sector of the linear three-Airy problem. The point $b=0$ lies on the degenerate locus and requires a separate treatment of the inverse problem.

\subsection{A fourth-order test of branch-velocity interpolation}
The simple third-order formulas suggest interpolating the field corrections from branch velocities. With $\zeta$ an indeterminate independent of $x$, the trial expression was
\begin{equation}
 \sum_{j=0}^{r-2}u_{j,1}\zeta^j
 \mathrel{\mathop=^{?}}\frac r{24}\partial_x
 \left[\sum_i\prod_{k\ne i}(\zeta-\alpha_k)
                         \partial_x\log v_i\right].
\end{equation}
It reproduces the checked second- and third-order answers. At fourth order, a family with leading $u_1=0$ and moving $u_2,u_0$ reveals a discrepancy in the constant coefficient. In this paragraph $u_2,u_0$ denote leading coefficient fields. The computed correction minus the trial correction is
\begin{equation}
 -\frac{6(u_0')^2u_2''+6u_0'u_0''u_2'
 -5u_0'u_2u_2'u_2''-2u_0'(u_2')^3-u_0''u_2(u_2')^2}
 {24u_0'(2u_0'-u_2u_2')}.
\end{equation}
The displayed term is the correction needed by the trial interpolation at fourth order. The higher coefficients agree in this family, and the difference vanishes when the higher field is constant. Detecting it therefore requires a moving background and a comparison of the lower fields. The remainder construction includes this contribution by determining every field correction independently.

\section{Further Ising and unitary data}\label{app:polynomials}
\subsection{The Ising genus-two and genus-three numerators}
The genus-two numerator used in the text is
\begin{align}\label{eq:IsingP2}
 P_2(p)={}&77p^{18}-539p^{16}+1982p^{14}-4098p^{12}+6556p^{10}\nonumber\\
 &-5060p^8+8010p^6+5642p^4+10831p^2+791.
\end{align}

The polynomial appearing in \eqref{eq:IsingW31} is
\begin{align}
 P_3(p)={}&53955p^{30}-701415p^{28}+4465420p^{26}-17959320p^{24}\nonumber\\
 &+51069755p^{22}-107773575p^{20}+175823400p^{18}-221591820p^{16}\nonumber\\
 &+227618305p^{14}-165273597p^{12}+136206860p^{10}+14026336p^8\nonumber\\
 &+126732801p^6+139051115p^4+51852480p^2+1534020.
\end{align}
With the normalization \eqref{eq:BBEscale}, comparison with the literature reproduces every coefficient of the Ising genus-three one-point function \cite[Section~6.3]{BergereBorotEynard2015}. In particular, this checks the lower poles as well as the leading behavior at ramification.

The homogeneous string equation also gives a useful independent coefficient recurrence. Write
\begin{equation}
 \mathfrak u=\sum_{g\geq0}c_g\h^{2g}x^{\gamma_g},\qquad
 \gamma_g=\frac{1-7g}{3},\qquad c_0=1,
\end{equation}
so that $b=-3\mathfrak u$. Expanding \eqref{eq:Isingstring} isolates the new coefficient as
\begin{align}
 12c_g={}&3\sum_{i+j=g-1}c_ic_j\gamma_j(\gamma_j-1)
 +\frac32\sum_{i+j=g-1}c_ic_j\gamma_i\gamma_j\nonumber\\
 &-\frac{\mathbf1_{g\geq2}}6c_{g-2}
        \prod_{j=0}^3(\gamma_{g-2}-j)
 -4\!\sum_{\substack{i+j+k=g\\ i,j,k<g}}c_ic_jc_k.
\end{align}
The indicator $\mathbf1_{g\geq2}$ is one when $g\geq2$ and zero otherwise. This recurrence follows from the specified string equation and establishes the coefficients for that solution. Reconstruction for a general leading family remains subject to the primitive obstruction tests.

\subsection{Two further multi-boundary outputs}\label{app:isingmulti}
The polynomial entering the complete genus-one two-boundary function is
\begin{align}\label{eq:IsingH}
 \mathsf H_\alpha(\xi,\eta)={}&69-27\alpha(\xi+\eta)
       +81(\xi^2+\eta^2)+4\xi\eta\nonumber\\
 &-40\alpha(\xi^3+\eta^3)-12\alpha\xi\eta(\xi+\eta)
       +60(\xi^4+\eta^4)+36\xi^2\eta^2.
\end{align}

The branch-pole basis keeps the four-boundary Ising answer manageable. With $L_\alpha(p)=(p-\alpha)^{-1}$ and $\alpha=\pm1$, the exact EO result is
\begin{align}\label{eq:IsingW04appendix}
 W_{0,4}(p_1,p_2,p_3,p_4)=\frac1{144}\Bigg\{
 &\sum_{\alpha=\pm1}\prod_{i=1}^4L_\alpha(p_i)^2
 \left[13-4\alpha\sum_{i=1}^4L_\alpha(p_i)
              +12\sum_{i=1}^4L_\alpha(p_i)^2\right]\nonumber\\
 &+\sum_{\substack{S\subset\{1,2,3,4\}\\|S|=2}}
       \prod_{i\in S}L_+(p_i)^2\prod_{j\notin S}L_-(p_j)^2
 \Bigg\}.
\end{align}
The second sum contains six terms. It separates the contributions involving both ramification points from those supported on a single branch. In particular,
\begin{equation}
 W_{0,4}(0,0,0,0)=\frac{10}{9}.
\end{equation}
The complete residue formula was checked against the matrix cycles at eight exact rational four-tuples. These evaluations establish agreement at the stated points; a full four-variable symbolic comparison remains open.

The matrix calculation also gives the genus-two specialization
\begin{equation}\label{eq:IsingW22appendix}
 W_{2,2}(0,p)=\frac{5\mathcal P_{22}(p)}{46656(p^2-1)^{12}},
\end{equation}
with
\begin{align}
 \mathcal P_{22}(p)={}&4746p^{22}-42714p^{20}+182247p^{18}-469857p^{16}\nonumber\\
 &+817272p^{14}-987816p^{12}+873134p^{10}-514938p^8\nonumber\\
 &+286814p^6+9346p^4+91083p^2+16795.
\end{align}
It obeys
\begin{equation}
 W_{2,2}(0,0)=\frac{83975}{46656},\qquad
 W_{2,2}(0,2)=\frac{412026425}{918330048}.
\end{equation}
A separate exact projector coefficient extraction reproduces \eqref{eq:IsingW22appendix}, and a local expansion checks the six most singular Laurent coefficients, those of $(p-\alpha)^{-12},\ldots,(p-\alpha)^{-7}$ at each $\alpha=\pm1$. The first external coordinate is fixed at the ordinary point $p_1=0$; the Laurent orders refer to the second coordinate approaching a ramification point. The local comparison uses the required subleading jets of the Ising curve, not merely the leading quadratic Airy model. Reproducing the rational function by a second projector extraction is an internal check of the one-variable specialization, while the six Laurent coefficients supply only a partial independent local comparison. They do not replace an EO derivation of the entire specialization or of the full two-variable function, which remains to be compared.

\subsection{The fourth- and fifth-order genus-two numerators}
For \eqref{eq:W21r4} and \eqref{eq:W21r5}, respectively, the complete polynomials are
\begin{align}
 N_4(p)={}&1365p^{28}-22188p^{26}+162195p^{24}-694163p^{22}\nonumber\\
 &+1906498p^{20}-3449435p^{18}+4034219p^{16}-2759855p^{14}\nonumber\\
 &+775586p^{12}+262164p^{10}+99400p^8-13776p^6\nonumber\\
 &+90720p^4-141120p^2+67200,
\end{align}
\begin{align}
 N_5(p)={}&13090p^{38}-332910p^{36}+3872150p^{34}-27240062p^{32}\nonumber\\
 &+129234637p^{30}-436758505p^{28}+1083697673p^{26}\nonumber\\
 &-2007954230p^{24}+2804943208p^{22}-2971620568p^{20}\nonumber\\
 &+2400320300p^{18}-1487844014p^{16}+715395340p^{14}\nonumber\\
 &-272524798p^{12}+84320108p^{10}-18533930p^8\nonumber\\
 &-610847p^6-1414105p^4+4387107p^2+127138.
\end{align}
These functions agree with independent EO calculations in their full rational form. Their comparison with the literature is through the standard string-pair construction; the displayed coefficients are obtained here by reconstruction and tested by residues.

A compact way to display additional two-boundary data is to fix the first argument away from ramification. At $p_1=3$,
\begin{align}
 W^{(4,5)}_{1,2}(3,q)&=\frac{S_4(q)}{10978063488q^6(q^2-2)^6},\\
 W^{(5,6)}_{1,2}(3,q)&=\frac{S_5(q)}{332167687500(q^4-3q^2+1)^6},
\end{align}
where
\begin{align}
 S_4(q)={}&123807473q^{16}+158015124q^{15}-1011778251q^{14}\nonumber\\
 &-1264120992q^{13}+3474719310q^{12}+4317033024q^{11}\nonumber\\
 &-5676187295q^{10}-7155164160q^9+4390289805q^8\nonumber\\
 &+6307238016q^7+1128752040q^6+851125688q^4\nonumber\\
 &-864014256q^2+762365520,
\end{align}
\begin{align}
 S_5(q)={}&5230869288q^{22}+8113077828q^{21}-70543049778q^{20}\nonumber\\
 &-110350837740q^{19}+404434456890q^{18}+642384284400q^{17}\nonumber\\
 &-1268332395550q^{16}-2065146046500q^{15}+2360696284710q^{14}\nonumber\\
 &+3994744533660q^{13}-2670664676359q^{12}-4773542768244q^{11}\nonumber\\
 &+1920744959259q^{10}+3629235740760q^9-899547473910q^8\nonumber\\
 &-1759182695400q^7+319782299925q^6+506200551300q^5\nonumber\\
 &-213829918090q^4-86685948840q^3+168169700253q^2\nonumber\\
 &+47774390328q+19894193387.
\end{align}
The letter $q$ here is only the second boundary's curve coordinate. The exact comparison uses the full two-variable expressions obtained from \eqref{eq:W12matrix}. The one-variable specializations displayed here present their coefficients compactly.

\subsection{All fields at orders six and seven}\label{app:higherfields}
These examples test the same inverse calculation with additional lower fields. At order six, with $U=x^{1/6}$, the reconstructed symmetric coefficients are
\begin{align}
 u_4&=-6U+\frac58\h^2U^{-12}
             +\frac{164255}{62208}\h^4U^{-25}+O(\h^6),\\
 u_2&=9U^2-\frac{55}{24}\h^2U^{-11}
             -\frac{90295}{10368}\h^4U^{-24}+O(\h^6),\\
 u_0&=-2U^3+\frac{31}{48}\h^2U^{-10}
             +\frac{22435}{7776}\h^4U^{-23}+O(\h^6),
\end{align}
with $u_3=u_1=0$ through the displayed orders. At order seven, with $U=x^{1/7}$, they are
\begin{align}
 u_5&=-7U+\frac78\h^2U^{-14}
             +\frac{8349}{1568}\h^4U^{-29}+O(\h^6),\\
 u_3&=14U^2-\frac{30}{7}\h^2U^{-13}
             -\frac{144763}{6272}\h^4U^{-28}+O(\h^6),\\
 u_1&=-7U^3+\frac{83}{28}\h^2U^{-12}
             +\frac{14115}{784}\h^4U^{-27}+O(\h^6),
\end{align}
with $u_4=u_2=u_0=0$ through those orders. The zero corrections are outputs of the all-field homogeneous fit.

At order six the genus-one and genus-two fits have, respectively, 29 equations for 24 unknowns and 59 equations for 54 unknowns, each with full column rank and equal augmented rank. At order seven the corresponding counts are 35 for 29 and 71 for 65. Independent fractional-power calculations verify the string commutators and the trace primitives. The full genus-one and genus-two one-boundary functions also agree with EO residues. Their amplitude polynomials are determined by the fields above and the primitive equations of Section~\ref{sec:reconstruction}.

\section{Exact data for dual and nonunitary projections}\label{app:NUdata}
This appendix makes the finite-order outputs explicit. All coefficients use the symmetric ordering, and all amplitudes are evaluated at $x=1$, so $z=p$. The two integers $(r,s)$ list the orders of $Q,P$ in that order. The field expansion is exactly \eqref{eq:homfields}; the dimensionless number $\beta_{j,g}$ multiplies $\h^{2g}U^{(r-j)/2-(r+s)g}$, with $U=x^{2/(r+s-1)}$. This convention is essential when comparing two rows with the same operator order and different partners.

\subsection{All field coefficients through genus two}
\begingroup\small\renewcommand{\arraystretch}{1.20}
\begin{longtable}{@{}ccrrr@{}}
\caption{Homogeneous symmetric-field coefficients in the normalization of \eqref{eq:homfields}.}\label{tab:allNUfields}\\
\toprule
$(r,s)$ & $j$ & $\beta_{j,0}$ & $\beta_{j,1}$ & $\beta_{j,2}$\\\midrule
\endfirsthead
\toprule $(r,s)$ & $j$ & $\beta_{j,0}$ & $\beta_{j,1}$ & $\beta_{j,2}$\\\midrule
\endhead
\bottomrule\endfoot

$(3,2)$ & $0$ & $0$ & $0$ & $0$\\

$(3,2)$ & $1$ & $-3$ & $\frac{1}{16}$ & $\frac{49}{1536}$\\

\addlinespace

$(4,3)$ & $0$ & $2$ & $- \frac{13}{54}$ & $- \frac{5713}{23328}$\\

$(4,3)$ & $1$ & $0$ & $0$ & $0$\\

$(4,3)$ & $2$ & $-4$ & $\frac{1}{6}$ & $\frac{1925}{11664}$\\

\addlinespace

$(5,4)$ & $0$ & $0$ & $0$ & $0$\\

$(5,4)$ & $1$ & $5$ & $- \frac{115}{128}$ & $- \frac{25631}{16384}$\\

$(5,4)$ & $2$ & $0$ & $0$ & $0$\\

$(5,4)$ & $3$ & $-5$ & $\frac{5}{16}$ & $\frac{17745}{32768}$\\

\addlinespace

$(6,5)$ & $0$ & $-2$ & $\frac{57}{100}$ & $\frac{173249}{75000}$\\

$(6,5)$ & $1$ & $0$ & $0$ & $0$\\

$(6,5)$ & $2$ & $9$ & $- \frac{21}{10}$ & $- \frac{56923}{10000}$\\

$(6,5)$ & $3$ & $0$ & $0$ & $0$\\

$(6,5)$ & $4$ & $-6$ & $\frac{1}{2}$ & $\frac{2618}{1875}$\\

\addlinespace

$(3,5)$ & $0$ & $9$ & $0$ & $\frac{1573}{36015}$\\

$(3,5)$ & $1$ & $9$ & $\frac{1}{7}$ & $- \frac{3289}{108045}$\\

\addlinespace

$(5,3)$ & $0$ & $90$ & $\frac{40}{49}$ & $\frac{2860}{21609}$\\

$(5,3)$ & $1$ & $45$ & $\frac{120}{49}$ & $- \frac{1847}{2401}$\\

$(5,3)$ & $2$ & $15$ & $0$ & $\frac{1573}{21609}$\\

$(5,3)$ & $3$ & $15$ & $\frac{5}{21}$ & $- \frac{3289}{64827}$\\

\addlinespace

$(5,2)$ & $0$ & $0$ & $0$ & $0$\\

$(5,2)$ & $1$ & $\frac{15}{2}$ & $- \frac{25}{36}$ & $- \frac{385}{864}$\\

$(5,2)$ & $2$ & $0$ & $0$ & $0$\\

$(5,2)$ & $3$ & $-5$ & $\frac{5}{36}$ & $\frac{35}{432}$\\

\addlinespace

$(2,5)$ & $0$ & $-2$ & $\frac{1}{18}$ & $\frac{7}{216}$\\

\addlinespace

\end{longtable}
\endgroup
For $(3,2)$ at genus three, the remaining coefficients are $\beta_{0,3}=0$, $\beta_{1,3}=1225/18432$. The full-column-rank homogeneous systems determine these coefficients and the primitive simultaneously. The reconstruction determines both the zero and nonzero entries in the positive-genus columns from the specified leading coefficients.

\subsection{Further genus-one boundary formulas}
The fifth-order Yang--Lee numerator in \eqref{eq:52W11} is
\begin{align}
 B_{52}(p)={}&16p^{14}-168p^{12}+672p^{10}-1908p^8
 +3780p^6\nonumber\\
 &-4914p^4+2592p^2+567.
\end{align}
The $(3,5)$ two-boundary specialization quoted in the text is
\begin{align}\label{eq:35W12slice}
 W^{(3,5)}_{1,2}(0,p)&=-\frac{A_{35}(p)}{441(p^2+3)^6},\nonumber\\
 A_{35}(p)&=16p^{10}+66p^9+144p^8+648p^7+648p^6+3564p^5\nonumber\\
 &\quad+891p^4+4536p^3-243p^2+22842p-6318.
\end{align}

\paragraph{$(5,4)$.}

\begin{equation} W^{(5,4)}_{1,1}(z)=-\frac{\mathcal A_{5,4}(z)}{32\bigl(z^{4} - 3 z^{2} + 1\bigr)^4}.\end{equation}

\begin{align*}
 \mathcal A_{5,4}(z)&=2z^{14}-17z^{12}+57z^{10}-84z^{8}\\
 &\quad+85z^{6}-84z^{4}+55z^{2}+9 .
\end{align*}

\paragraph{$(6,5)$.}

\begin{equation} W^{(6,5)}_{1,1}(z)=-\frac{\mathcal A_{6,5}(z)}{120\bigl(z^{5} - 4 z^{3} + 3 z\bigr)^4}.\end{equation}

\begin{align*}
 \mathcal A_{6,5}(z)&=10z^{18}-121z^{16}+601z^{14}-1562z^{12}\\
 &\quad+2317z^{10}-1872z^{8}+387z^{6}+594z^{4}\\
 &\quad-243z^{2}+81 .
\end{align*}

\paragraph{$(5,3)$.}

\begin{equation} W^{(5,3)}_{1,1}(z)=-\frac{\mathcal A_{5,3}(z)}{21\bigl(z^{4} + 9 z^{2} + 6 z + 9\bigr)^4}.\end{equation}

\begin{align*}
 \mathcal A_{5,3}(z)&=z^{14}+27z^{12}+15z^{11}+288z^{10}\\
 &\quad+270z^{9}+1674z^{8}-162z^{7}+5670z^{6}\\
 &\quad-7641z^{5}+6318z^{4}-486z^{3}-31995z^{2}\\
 &\quad+86265z+26244 .
\end{align*}

\subsection{The full genus-two one-boundary functions}
For compactness define the integer-coefficient ramification polynomial $\mathscr D_{r,s}(z)$ and the constants $\epsilon_{r,s}\in\{+1,-1\}$ and $C_{r,s}>0$ by the following table. Then the complete amplitudes are
\begin{equation}\label{eq:NUW21appendix}
 W^{(r,s)}_{2,1}(z)=\epsilon_{r,s}
 \frac{\mathcal N_{r,s}(z)}{C_{r,s}\,\mathscr D_{r,s}(z)^{10}}.
\end{equation}
The polynomial $\mathscr D_{5,2}$ is twice the monic polynomial used in the quotient-algebra calculation; this factor is already included in $C_{5,2}$. The numerators below have positive leading coefficients.
\begin{center}\small
\begin{tabular}{@{}cccc@{}}\toprule
$(r,s)$ & $\mathscr D_{r,s}(z)$ & $\epsilon_{r,s}$ & $C_{r,s}$\\\midrule

$(3,2)$ & $\left(z - 1\right) \left(z + 1\right)$ & $-1$ & $2304$\\

$(4,3)$ & $z \left(z^{2} - 2\right)$ & $-1$ & $31104$\\

$(5,4)$ & $\left(z^{2} - z - 1\right) \left(z^{2} + z - 1\right)$ & $-1$ & $40960$\\

$(6,5)$ & $z \left(z - 1\right) \left(z + 1\right) \left(z^{2} - 3\right)$ & $-1$ & $720000$\\

$(3,5)$ & $z^{2} + 3$ & $1$ & $46305$\\

$(5,3)$ & $z^{4} + 9 z^{2} + 6 z + 9$ & $1$ & $46305$\\

$(5,2)$ & $2 z^{4} - 6 z^{2} + 3$ & $-1$ & $6480$\\

$(2,5)$ & $z$ & $-1$ & $51840$\\

\bottomrule\end{tabular}\end{center}

\paragraph{$(3,2)$.}

\begin{align*}
 \mathcal N_{3,2}(z)&=7z^{18}-49z^{16}+168z^{14}-287z^{12}\\
 &\quad+413z^{10}+270z^{8}+3470z^{6}+16873z^{4}\\
 &\quad+9014z^{2}+361 .
\end{align*}

\paragraph{$(4,3)$.}

\begin{align*}
 \mathcal N_{4,3}(z)&=385z^{28}-6545z^{26}+50335z^{24}-229315z^{22}\\
 &\quad+684730z^{20}-1397590z^{18}+1978315z^{16}-1894915z^{14}\\
 &\quad+1216090z^{12}-89820z^{10}+219560z^{8}-412560z^{6}\\
 &\quad+572640z^{4}-409920z^{2}+120960 .
\end{align*}

\paragraph{$(5,4)$.}

\begin{align*}
 \mathcal N_{5,4}(z)&=1365z^{38}-35838z^{36}+431850z^{34}-3160943z^{32}\\
 &\quad+15682128z^{30}-55736390z^{28}+146291700z^{26}-288215160z^{24}\\
 &\quad+429253517z^{22}-483501562z^{20}+409871060z^{18}-258093851z^{16}\\
 &\quad+117050000z^{14}-34769452z^{12}+1683282z^{10}+9662540z^{8}\\
 &\quad-10262939z^{6}+2697860z^{4}+1607088z^{2}+30597 .
\end{align*}

\paragraph{$(6,5)$.}

\begin{align*}
 \mathcal N_{6,5}(z)&=52360z^{48}-1855240z^{46}+30637600z^{44}-313071648z^{42}\\
 &\quad+2216533028z^{40}-11537330300z^{38}+45722115252z^{36}-140962250740z^{34}\\
 &\quad+342541057393z^{32}-660696038993z^{30}+1013689831750z^{28}-1234391656765z^{26}\\
 &\quad+1185115217910z^{24}-886877108693z^{22}+508733534238z^{20}-218754570465z^{18}\\
 &\quad+68541714864z^{16}-14947642395z^{14}+2212396578z^{12}-1062994023z^{10}\\
 &\quad+1349692470z^{8}-805646331z^{6}+355628070z^{4}-91165095z^{2}\\
 &\quad+10333575 .
\end{align*}

\paragraph{$(3,5)$.}

\begin{align*}
 \mathcal N_{3,5}(z)&=143z^{18}-429z^{17}+3003z^{16}-10296z^{15}\\
 &\quad+38610z^{14}-108108z^{13}+286119z^{12}-822393z^{11}\\
 &\quad+1622835z^{10}-3393495z^{9}+7020999z^{8}-17272197z^{7}\\
 &\quad+17748234z^{6}-22173993z^{5}+104538600z^{4}-184718394z^{3}\\
 &\quad-55066473z^{2}+150220656z-12203460 .
\end{align*}

\paragraph{$(5,3)$.}

\begin{align*}
 \mathcal N_{5,3}(z)&=143z^{38}-429z^{37}+11583z^{36}\\
 &\quad-27456z^{35}+411840z^{34}-764478z^{33}\\
 &\quad+8548659z^{32}-12289563z^{31}+115913430z^{30}\\
 &\quad-130064940z^{29}+1082465370z^{28}-1034666055z^{27}\\
 &\quad+7044805710z^{26}-7517892690z^{25}+30124339425z^{24}\\
 &\quad-57110349159z^{23}+56160675027z^{22}-413550550674z^{21}\\
 &\quad-273585927402z^{20}-2459688358545z^{19}-2932377729309z^{18}\\
 &\quad-11280449682333z^{17}-14017084155999z^{16}-40243802672268z^{15}\\
 &\quad-43402010566140z^{14}-133124718129759z^{13}-87574565232648z^{12}\\
 &\quad-539149891984704z^{11}-126814175650800z^{10}-2026336828198635z^{9}\\
 &\quad-691499924267913z^{8}-3994584925347981z^{7}-3815224363043253z^{6}\\
 &\quad-498922289956710z^{5}-5075994056387940z^{4}+5950430032599900z^{3}\\
 &\quad+7390912812550785z^{2}+598396199336820z-237104739958959 .
\end{align*}

\paragraph{$(5,2)$.}

\begin{align*}
 \mathcal N_{5,2}(z)&=32256z^{38}-883712z^{36}+11206400z^{34}-87384320z^{32}\\
 &\quad+469347840z^{30}-1842356736z^{28}+5470470912z^{26}-12545740800z^{24}\\
 &\quad+22493113920z^{22}-31727669760z^{20}+35574981408z^{18}-33745592736z^{16}\\
 &\quad+31044785760z^{14}-23462252640z^{12}-4740541200z^{10}+41068500768z^{8}\\
 &\quad-40171074066z^{6}+7325618940z^{4}+5400326295z^{2}+117802755 .
\end{align*}

\paragraph{$(2,5)$.}

\begin{align*}
 \mathcal N_{2,5}(z)&=252z^{8}+656z^{6}+1390z^{4}+2030z^{2}\\
 &\quad+1575 .
\end{align*}

Every function in \eqref{eq:NUW21appendix} has a full independent EO comparison. We include the numerators so that all coefficients can be compared. The construction of the longer two-variable $W_{1,2}$ functions and the unreduced primitives is specified by \eqref{eq:homprimitive} and \eqref{eq:W12matrix}.

\section{Conventions for comparison with the literature}\label{app:published}
To compare with the literature, we need to put both calculations in the same variables and ordering convention. This appendix gives the required maps for the full operators and their amplitudes. We begin with the Berg\`ere--Borot--Eynard Ising pair:
\begin{align}\label{eq:BBEIsingops}
 Q_B&=D_B^3-3\mathfrak uD_B-\frac32\h_B\mathfrak u_{t_B},\nonumber\\
 P_B&=D_B^4-4\mathfrak uD_B^2-4\h_B\mathfrak u_{t_B}D_B
      +2\mathfrak u^2-\frac53\h_B^2\mathfrak u_{t_Bt_B}.
\end{align}
With $t_B=4x$ and $\h_B=4\h$, their genus-one one-boundary expression is
\begin{equation}
 \omega^B_{1,1}=-\frac1{576}\left[
 \frac{7+7p+3p^2}{(p+1)^4}
 +\frac{7-7p+3p^2}{(p-1)^4}\right]\dd p.
\end{equation}
Combining the fractions and multiplying by four gives \eqref{eq:IsingW11}. There is a related sign in the free-energy convention. BBE uses $\mathfrak u=\h_B^2\partial_{t_B}^2\log Z$, so $b=-3\h^2\partial_x^2\log Z$. Equation~\eqref{eq:Floc}, on the other hand, has $B=3\partial_x^2F_1^{\rm loc}$ and yields $F_1^{\rm loc}=-\log x/24$. Thus the coefficient of genus one in $\log Z$ has the opposite sign, up to an additive constant. The two expressions agree once the field convention is included.

The parameters with a subscript C below belong to Chan, Irie, Niedner and Yeh \cite{ChanIrieNiednerYeh2016}. To keep their literal operator names distinct from ours, write their Baker--Akhiezer system as
\begin{equation*}
 \zeta\psi=\mathsf P_{\mathrm C}\psi,\qquad
 g_{\mathrm C}\partial_\zeta\psi=\mathsf Q_{\mathrm C}\psi,
 \qquad [\mathsf P_{\mathrm C},\mathsf Q_{\mathrm C}]=g_{\mathrm C}.
\end{equation*}
Their $\mathsf P_{\mathrm C}$ is the spectral operator, whereas our spectral operator is called $Q$. Their $\mathsf Q_{\mathrm C}$ includes the overall coupling $\beta_{p,q}$ multiplying the differential polynomial printed in the Baker--Akhiezer equation. We retain this coupling in the operator dictionary below rather than silently removing it. The symbols $p,q$ in $\beta_{p,q}$ label the source paper's minimal model; they are not our curve coordinate or a new pair of variables.

\subsection{The third-/fifth-order nonunitary pair}
Let $a,b$ be the fields in \eqref{eq:35fields}, and write $\gamma_{\mathrm C}=\beta_{3,5}\ne0$. The coefficient fields and the complete operators in the source convention are
\begin{equation}\label{eq:Chan35dict}
 \begin{aligned}
 U_{2,\mathrm C}&=4b,\qquad U_{3,\mathrm C}=4a+2\h b',\\
 \mathsf P_{\mathrm C}&=4Q,\qquad
 \mathsf Q_{\mathrm C}=16\gamma_{\mathrm C}P.
 \end{aligned}
\end{equation}
In particular, the derivative term in $U_{3,\mathrm C}$ is required by the conversion from symmetric to left ordering. Its coefficient is fixed by the ordering. The coordinates and dispersion parameters are
\begin{equation}
 t_{\mathrm C}=25920\gamma_{\mathrm C}x,\qquad
 g_{\mathrm C}=25920\gamma_{\mathrm C}\h,\qquad
 \mu_{\mathrm C}=0.
\end{equation}
Thus $g_{\mathrm C}\partial_{t_{\mathrm C}}=D$, and the spectral variables are related by $\zeta=4E$. The commutator provides a direct check of both the operator order and the coupling factor:
\begin{equation*}
 [\mathsf P_{\mathrm C},\mathsf Q_{\mathrm C}]
 =64\gamma_{\mathrm C}[Q,P]
 =25920\gamma_{\mathrm C}\h=g_{\mathrm C},
\end{equation*}
since our pair satisfies $[P,Q]=-405\h$. The parameter $\mu_{\mathrm C}$ specifies their background. Our $\mu$ denotes the endpoint, set to one in the displayed amplitudes.

It is useful to follow both field equations explicitly through the conversion. Let $\mathscr D=\h\partial_x$ when it acts on a coefficient function, and abbreviate $U_2=U_{2,\mathrm C}$, $U_3=U_{3,\mathrm C}$. The two differential polynomials in their convention are
\begin{align}
 C_1={}&60U_3(\mathscr D U_2-U_3)+20U_2\mathscr D^2U_2
          +\frac53U_2^3+16\mathscr D^4U_2,\\
 C_2={}&20(2\mathscr D U_2-3U_3)\mathscr D U_3
          +5U_3(8\mathscr D^2U_2+U_2^2)\nonumber\\
       &+20U_2\mathscr D^2U_3+16\mathscr D^4U_3.
\end{align}
After \eqref{eq:Chan35dict}, $-C_1/960$ is the residual of the first equation in \eqref{eq:35strings}, with its right-hand side moved to the left, and $(C_2-\mathscr D C_1/2)/320$ is the second left-hand side before subtracting $729x$. Their equations on the chosen slice therefore become exactly \eqref{eq:35strings}. Under this conversion, the odd derivative terms combine across the two equations and cancel. The corresponding complete fifth-order operator also agrees with the separately reconstructed $(5,3)$ member.

\subsection{Yang--Lee and the second-order partner}
For the potential $\mathfrak u$ in \eqref{eq:YLstring}, put $\gamma_{\mathrm C}=\beta_{2,5}\ne0$. The literal source-operator dictionary is
\begin{equation}
 U_{2,\mathrm C}=-4\mathfrak u,\qquad
 \mathsf P_{\mathrm C}=2P,\qquad
 \mathsf Q_{\mathrm C}=16\gamma_{\mathrm C}Q.
\end{equation}
Here their spectral operator $\mathsf P_{\mathrm C}$ is second order, while our $Q$ in Section~\ref{sec:nonunitary} is fifth order. This exchange of projection is separate from the numerical coefficient normalization. The coordinate and dispersion changes are
\begin{equation}
 t_{\mathrm C}=160\gamma_{\mathrm C}x,\qquad
 g_{\mathrm C}=160\gamma_{\mathrm C}\h,\qquad
 \mu_{\mathrm C}=0.
\end{equation}
They obey $g_{\mathrm C}\partial_{t_{\mathrm C}}=D$, and our $[P,Q]=5\h$ gives
\begin{equation*}
 [\mathsf P_{\mathrm C},\mathsf Q_{\mathrm C}]
 =32\gamma_{\mathrm C}[P,Q]
 =160\gamma_{\mathrm C}\h=g_{\mathrm C}.
\end{equation*}
Their integrated differential polynomial becomes
\begin{equation}
 -5(\mathscr D U_2)^2-10U_2\mathscr D^2U_2
          -5U_2^3-2\mathscr D^4U_2=320x.
\end{equation}
Inserting $U_2=-4\mathfrak u$ and dividing by 320 gives \eqref{eq:YLstring}. The fifth-order coefficients agree with the operator in their Eq.~(B.28), and the displayed differential polynomial is their Eq.~(B.29) on this slice \cite{ChanIrieNiednerYeh2016}. These comparisons follow directly from the operator and coordinate dictionary.

Both nonunitary examples use a particular homogeneous zero-$\mu_{\mathrm C}$ background. The one-scale polynomial forms displayed here describe these chosen backgrounds; other conformal-times choices require their own background identification. The nonzero classical bracket, the explicit operator dictionary and the disk normalization together specify the comparison.

\section{Higher-Airy calculations and comparisons}\label{app:airydata}
\subsection{One-point comparison with the literature through genus six}
We first state the one-point formula from the literature, then compare it with the scalar coefficients. Define
\begin{equation}
 t_i^{\rm LVX}=-\frac{\binom r{2i}}{(2i+1)4^i},\qquad
 |\mathbf m|=\sum_i i m_i,\qquad
 \|\mathbf m\|=\sum_i m_i,\qquad \mathbf m!=\prod_i m_i!.
\end{equation}
The symbols $t_i^{\rm LVX}$ denote the numerical coefficients in this closed formula. We reserve $t_m$ for the hierarchy times in \eqref{eq:hierarchy}. For the allowed label fixed by \eqref{eq:onepointselection}, Proposition~5.4 of \cite{LiuVakilXu2011} gives
\begin{equation}\label{eq:LVX}
 \langle\tau_{d,\nu}\rangle_g=
 \frac{(-1)^g}{r^g\Gamma(1-(\nu+1)/r)}
 \sum_{|\mathbf m|=g}
 \Gamma\!\left(\|\mathbf m\|-\frac{2g-1}{r}\right)
 \frac{\prod_i(t_i^{\rm LVX})^{m_i}}{\mathbf m!}.
\end{equation}
Only $i\leq\lfloor r/2\rfloor$ occur. For $r=4,5$, the sum is over $m_1+2m_2=g$, and the gamma ratios reduce to rational products. The comparison can therefore be made exactly.

After applying \eqref{eq:rspindictionary}, the scalar concomitant calculation through $\h^{12}$ gives Table~\ref{tab:Airy}. Each nonzero entry agrees with \eqref{eq:LVX}, and the ODE also gives the zero required by the narrow-sector rule. The calculation therefore reproduces these known intersection numbers in the scalar normalization.
\begin{table}[htbp]
\centering\small
\begin{tabular}{@{}ccccc@{}}\toprule
$g$&$(d,\nu)$, $r=4$&$4$-spin coefficient&$(d,\nu)$, $r=5$&$5$-spin coefficient\\\midrule
1&$(1,0)$&$1/8$&$(1,0)$&$1/6$\\
2&$(3,2)$&$3/2560$&$(3,2)$&$11/3600$\\
3&$(6,0)$&$3/20480$&$(5,4)$&$0$ (forbidden)\\
4&$(8,2)$&$77/39321600$&$(8,1)$&$341/25920000$\\
5&$(11,0)$&$19/104857600$&$(10,3)$&$161/777600000$\\
6&$(13,2)$&$59/33554432000$&$(13,0)$&$3397/93312000000$\\\bottomrule
\end{tabular}
\caption{One-boundary higher-Airy coefficients compared with the literature using \eqref{eq:LVX}. The field label $\nu=4$ at $r=5,g=3$ is outside the narrow range and is shown only to make the selection-rule zero explicit.}\label{tab:Airy}
\end{table}

\subsection{Additional boundaries as a normalization check}
The matrix calculation at $r=4,5$ gives the complete genus-one two-boundary functions
\begin{align}\label{eq:AiryW12}
 W^{\mathrm A,4}_{1,2}(p_1,p_2)
 &=\frac1{128}\left[
 45\left(\frac1{p_1^2p_2^{10}}+\frac1{p_1^{10}p_2^2}\right)
 -21\left(\frac1{p_1^4p_2^8}+\frac1{p_1^8p_2^4}\right)
 +\frac{25}{p_1^6p_2^6}\right],\nonumber\\
 W^{\mathrm A,5}_{1,2}(p_1,p_2)
 &=\frac1{25}\left[
 11\left(\frac1{p_1^2p_2^{12}}+\frac1{p_1^{12}p_2^2}\right)
 -9\left(\frac1{p_1^4p_2^{10}}+\frac1{p_1^{10}p_2^4}\right)\right.\nonumber\\
 &\hspace{1.25cm}\left.
 -8\left(\frac1{p_1^5p_2^9}+\frac1{p_1^9p_2^5}\right)
 +\frac6{p_1^7p_2^7}\right].
\end{align}
All energy-coincidence denominators cancel. The string and dilaton equations fix the coefficients with labels $(2,0),(0,0)$ and $(1,0),(1,0)$ from $\langle\tau_{1,0}\rangle_1=(r-1)/24$. The remaining coefficients follow from genus-one gluing and Witten's four-primary formula, as explained below. Thus the independent intersection-theoretic checks cover the complete functions.

The following subsections give the full $W_{0,3}$ and $W_{0,4}$ at these scalar orders. Their comparison tests the vanishing of lower forbidden powers of $\h$ in the cycle sums, as well as the final Laurent coefficients. The larger $W_{2,2}$ and $W_{3,2}$ outputs have only specified coefficient subsets checked independently. For example,
\begin{equation}\label{eq:Airyhigherchecks}
 \langle\tau_{4,1}\tau_{0,1}\rangle_2^{4\text{-spin}}=\frac1{320},\quad
 \langle\tau_{4,1}\tau_{0,1}\rangle_2^{5\text{-spin}}=\frac7{1200},\quad
 \langle\tau_{6,3}\tau_{0,1}\rangle_3^{5\text{-spin}}=\frac{11}{108000}
\end{equation}
agree with Corollary~5.5 of \cite{LiuVakilXu2011}. The last value gives a nonzero five-spin genus-three two-point coefficient. At the same genus, the one-point coefficient vanishes by its selection rule. These examples show how the allowed coefficients depend on the number of boundaries.

\subsection{A finite symmetric basis}
All formulas in this appendix have $\lambda=1$. For an exponent list $(k_1,\ldots,k_n)$, define
\begin{equation}
 m_{k_1,\ldots,k_n}(\mathbf p^{-1})
 =\sum_{\text{distinct permutations }\sigma}
                \prod_{i=1}^n p_i^{-k_{\sigma(i)}}.
\end{equation}
The sum contains one term for each distinct exponent assignment, with repeated assignments counted once. This convention fixes the numerical coefficients in the following formulas.

The three-boundary sphere functions are
\begin{align}
 W^{\mathrm A,4}_{0,3}&=-\frac14(3m_{4,2,2}+4m_{3,3,2}),\\
 W^{\mathrm A,5}_{0,3}&=-\frac15(4m_{5,2,2}+6m_{4,3,2}+8m_{3,3,3}).
\end{align}
They follow from the primary three-point selection rule
$\langle\tau_{0,\nu_1}\tau_{0,\nu_2}\tau_{0,\nu_3}\rangle_0
=\delta_{\nu_1+\nu_2+\nu_3,r-2}$ and the dictionary \eqref{eq:rspindictionary}. The matrix three-cycle yields the same finite Laurent functions.

At four boundaries the complete results are
\begin{align}
 W^{\mathrm A,4}_{0,4}={}&\frac1{16}\bigl(
 21m_{8,2,2,2}+24m_{7,3,2,2}+15m_{6,4,2,2}\nonumber\\
 &\hspace{2.2cm}+20m_{6,3,3,2}-36m_{4,4,3,3}\bigr),
\end{align}
\begin{align}
 W^{\mathrm A,5}_{0,4}={}&\frac1{25}\bigl(
 36m_{10,2,2,2}+48m_{9,3,2,2}+42m_{8,4,2,2}\nonumber\\
 &+56m_{8,3,3,2}+24m_{7,5,2,2}+36m_{7,4,3,2}
       +48m_{7,3,3,3}\nonumber\\
 &-64m_{5,5,3,3}-72m_{5,4,4,3}-162m_{4,4,4,4}\bigr).
\end{align}
The independent primary comparison uses Witten's four-point formula, quoted in \cite{LiuVakilXu2011}. When $\sum_i\nu_i=2r-2$,
\begin{equation}
 \left\langle\prod_{i=1}^4\tau_{0,\nu_i}\right\rangle_0
 =\frac1r\min_i\{\nu_i,r-1-\nu_i\}.
\end{equation}
The remaining four-point terms have one descendant of degree one and $\sum_i\nu_i=r-2$; their coefficient is one by the genus-zero descendant relation. These two cases account for all monomials displayed above, including their different signs under \eqref{eq:rspindictionary}.

\subsection{The genus-one two-point comparison}
The string and dilaton equations reduce two coefficients in \eqref{eq:AiryW12} to the one-point invariant $(r-1)/24$ \cite{Witten1993,LiuVakilXu2011}. The remaining terms are fixed by the genus-one gluing relation
\begin{equation}
 \langle\tau_{1,a}\tau_{0,b}\rangle_1
 =\frac1{24}\sum_{c=0}^{r-2}
 \langle\tau_{0,a}\tau_{0,b}\tau_{0,c}\tau_{0,r-2-c}\rangle_0,
 \qquad a+b=r.
\end{equation}
Here $a,b,c$ label the spin sectors in the intersection formula. In these cases the separating contribution vanishes by the dimension constraint. The four-primary formula then gives $1/96$ for $r=4,a=b=2$, and $1/60$ for $r=5,(a,b)=(2,3)$ or $(3,2)$. Applying the multifactorials and the power of $-r$ in \eqref{eq:rspindictionary} reproduces every term of the matrix two-point functions. This comparison is sensitive both to the odd intermediate projector coefficients and to the descendant-dependent normalization.

\subsection{Higher two-point functions and the checked subset}
The matrix two-cycle gives the full genus-two functions
\begin{align}
 W^{\mathrm A,4}_{2,2}={}&\frac1{32768}\bigl(
 -39501m_{20,2}-96768m_{19,3}-37791m_{18,4}\nonumber\\
 &-31185m_{16,6}-75264m_{15,7}-33579m_{14,8}
 -33957m_{12,10}-84480m_{11,11}\bigr),
\end{align}
\begin{align}
 W^{\mathrm A,5}_{2,2}={}&\frac1{3125}\bigl(
 -9867m_{24,2}-18326m_{23,3}-9702m_{22,4}\nonumber\\
 &-7722m_{19,7}-14161m_{18,8}-8272m_{17,9}
 -8437m_{14,12}-15876m_{13,13}\bigr).
\end{align}
The string and dilaton equations check the coefficients that reduce to the one-point functions. Corollary~5.5 of Liu, Vakil and Xu supplies the additional coefficient checks listed in \eqref{eq:Airyhigherchecks}. The independent checks cover these specified subsets of coefficients at genera two and three. The full two-point functions are outputs of the matrix algorithm; the remaining coefficients await independent verification.

\begingroup
\raggedright
\bibliographystyle{unsrtnat}
\bibliography{main}
\endgroup
\end{document}